\documentclass[10pt,journal,compsoc]{IEEEtran}
\usepackage{cite}
\usepackage{amsmath,amsfonts}
\usepackage{array}
\usepackage{subfig}
\usepackage{textcomp}
\usepackage{stfloats}
\usepackage{url}
\usepackage{graphicx}
\usepackage[font=footnotesize]{caption}
\usepackage{float}
\usepackage{bm}
\usepackage{xcolor}
\usepackage{tabulary}
\usepackage{blindtext}
\usepackage{amsthm}  
\usepackage{amssymb}
\usepackage{mathrsfs}
\usepackage{epstopdf}

\newcommand{\vect}[1]{\boldsymbol{#1}}

\DeclareMathOperator{\sgn}{\text{sgn}}

\DeclareMathOperator*{\argmax}{arg\,max}

\newtheorem{remark}{Remark}
\newtheorem{theorem}{Theorem}
\newtheorem{lemma}{Lemma}
\newtheorem{assumption}{Assumption}

\newtheorem{definition}{Definition}

\newtheorem{proposition}{Proposition}

\newtheorem{example}{Example}

\def\bpmat{\begin{pmatrix}}
\def\epmat{\end{pmatrix}}
\def\bbmat{\begin{bmatrix}}
\def\ebmat{\end{bmatrix}}
\def\be{\begin{equation}}
\def\ee{\end{equation}}
\def\bc{\begin{cases}}
\def\ec{\end{cases}}
\def\bes{\begin{equation*}}
\def\ees{\end{equation*}}
\def\gA{{\cal A}}

\def\gG{{\cal G}}

\def\gV{{\cal V}}

\def\sgn{\text{sgn}}

\def\bla{\color{black}}

\def\bf{\textbf{}}
\begin{document}
\title{Coevolution of Actions and Opinions in Networks of Coordinating and Anti-Coordinating Agents}

\author{Hong Liang, Mengbin Ye, \IEEEmembership{Member, IEEE}, Lorenzo Zino, \IEEEmembership{Senior Member, IEEE}, and Weiguo Xia, \IEEEmembership{Member, IEEE}\IEEEcompsocitemizethanks{\IEEEcompsocthanksitem H. Liang and W. Xia are with the Key Laboratory of Intelligent Control and Optimization for Industrial Equipment of the Ministry of Education and the School of Control Science and Engineering, Dalian University of Technology, Dalian, China. (e-mail: hongliang56@mail.dlut.edu.cn, wgxiaseu@dlut.edu.cn). M. Ye is with the School of Computer and Mathematical Sciences, University of Adelaide, Adelaide, ustralia (e-mail: ben.ye@adelaide.edu.au). L. Zino is with the Department of Electronics and Telecommunications, Politecnico di Torino, Torino, Italy (e-mail: lorenzo.zino@polito.it). This work was supported in part by the National Natural Science Foundation of China (61973051, 62122016), and LiaoNing Science and Technology Program (2023JH2/101700361). Corresponding author: Weiguo Xia}}

\markboth{Preprint under review}%
{Liang \MakeLowercase{\textit{et al.}}: shorttitle}

\IEEEtitleabstractindextext{
\begin{abstract}
In this paper, we investigate the dynamics of coordinating and anti-coordinating agents in a coevolutionary model for actions and opinions. In the model, the individuals of a population interact on a two-layer network, sharing their opinions and observing others' action, while revising their own opinions and actions according to a game-theoretic mechanism, grounded in the social psychology literature. First, we consider the scenario of coordinating agents, where convergence to a Nash equilibrium (NE) is guaranteed. We identify conditions for reaching consensus configurations and establish regions of attraction for these equilibria. Second, we study networks of anti-coordinating agents. In this second scenario, we prove that all trajectories converge to a NE by leveraging potential game theory. Then, we establish analytical conditions on the network structure and model parameters to guarantee the existence of consensus and polarized equilibria, characterizing their regions of attraction. 
\end{abstract}
}

\maketitle
\IEEEpeerreviewmaketitle
\IEEEraisesectionheading{\section{Introduction}\label{sec:intro}} 

\IEEEPARstart{O}{ver} the past decades, extensive research has been conducted on dynamic models of opinion formation across multiple disciplines
~\cite{french1956_socialpower,hegselmann2002opinion,biswas2009model,ding2010evolutionary,friedkin2011social_book,proskurnikov2017tutorial,ravazzi2021learning,Loy2022}.
These models have been instrumental in understanding opinion formation processes and social power,  further analyzing complex social phenomena, such as the emergence of consensus, oscillations and disagreement, or polarization~\cite{friedkin1990social,altafini2012consensus,bolzern2019opinion,acemoglu2011opinion, Velarde2020Polarization,shi2022,DePasquale2022,Bizyaeva2023}. Moreover, they have shown significant potential for applications in different areas, e.g., marketing~\cite{godes2009firm}, advertising~\cite{candia2009advertising}, recommendation systems~\cite{castro2017opinion,tang2013social}, and epidemic mitigation~\cite{She2022,Xu2024}. 
Meanwhile, evolutionary game theory has emerged as a powerful framework for realistically predicting how individuals make decisions from a limited set of possible actions and how these actions are revised due to social interactions~\cite{Jackson2015, young2011dynamics, montanari2010spread_innovation, riehl2018survey, flache2017opdyn_survey,ye2021nat}.

Research from social psychology, supported by empirical data, shows a strong intertwining between an individual's opinion and their actions. These studies highlight that within a group, sharing opinions can influence people’s behaviors, and in turn, an individual's opinion may be shaped by observing the behaviors of others~\cite{henry2003beyond,haidt2001emotional,gavrilets2017collective}. 
Mathematical models have been proposed to capture the nontrivial interdependency between opinion and actions~\cite{martins2009opinion,martins2008continuous,chowdhury2016,ceragioli2018,varma2017}, but typically assume that an agent's action is only directly influenced by their own opinion, neglecting the presence of other mechanisms, e.g., social pressure. Recently, a novel paradigm has been proposed, in which actions and opinions coevolve on a two-layer network model~\cite{zino2020two,zino2020coevolutionary,aghbolagh2023coevolutionary}.

In particular, Dehghani Aghbolagh \textit{et al.} proposed a coevolutionary model for opinions and actions in an evolutionary game-theoretic form~\cite{aghbolagh2023coevolutionary}. The proposed model accounts for social influence due to opinion sharing between individuals, social pressure due to observing the actions of others, and an individual's tendency to act in consistency with their opinion. 
Potential game theory was leveraged to prove the global convergence of the dynamics to a Nash equilibrium (NE)~\cite{monderer1996potential}, and the existence and stability of polarized equilibria were studied.
The model of Dehghani Aghbolagh \textit{et al.} relies on the assumption that individuals are always faced with a situation of strategic complements, and thus tend to coordinate with others on the same action.  However, this is not the case in many real-life scenarios, where strategic substitutes are often present~\cite{Bramoull2007}. For instance, when a similar service is offered by different providers that are subject to congestion (e.g, different paths to reach a destination), the provider that is less congested becomes more attractive. Hence, the observation of others' actions yields a negative feedback effect, whereby agents tend to anti-coordinate with others to reduce congestion.

In this paper, we build on the coevolutionary model proposed in~\cite{aghbolagh2023coevolutionary}, and extend the model to account for the presence of coordinating or anti-coordinating agents to address the aforementioned limitation. Besides proposing a general formulation of the coevolutionary model, which can be used either in the presence of strategic complements or substitutes, our contribution is threefold, centered on the two emergent phenomena of consensus and polarization. First, we focus on the scenario of coordinating agents, and investigate how actions in coordination games enable social groups to reach a consensus despite their diversity, with the entire population selecting and supporting a certain action. Specifically, we establish analytical conditions on the network structure (depending on the model parameters) for the existence of consensus equilibria, and for the characterization of their basins of attraction. Second, we focus on the scenario of anti-coordinating agents, and we use the theory of potential games to prove that the dynamics converge to a NE~\cite{monderer1996potential}. Third, building on this result, we study two types of possible equilibria of interests ---consensus and polarized states--- and we establish conditions for their existence and convergence. Together with the characterization of polarized equilibria for the scenario with coordinating agents from~\cite{aghbolagh2023coevolutionary}, our work thoroughly explores consensus and polarization in networks with coevolving actions and opinions, in the presence of anti-coordinating agents.


The rest of this paper is organized as follows. Section~\ref{sec:prelim} provides the notational and mathematical preliminaries. In Section~\ref{sec:model}, we introduce the model, and derive its closed-form update rule. 
The main results are presented in Sections~\ref{sec: results_coo} and~\ref{sec: results_anti_coord}, respectively. Section~\ref{sec:conclusion} concludes the paper.
\bla
\section{Preliminaries}\label{sec:prelim}
\subsection{Notation and Graph Theory}

We use bold lowercase font to denote vectors (e.g., $\vect{v}\in\mathbb R^n$), while $v_i$ denotes its $i$th entry. The $n$-dimensional vectors of all $1$s and all $0$s are denoted as $\vect{1}_n$ and $\vect 0_n$, respectively. We use Roman uppercase font to denote matrices (e.g., $M\in\mathbb R^{n\times m}$), denoting by $m_{ij}$ the $j$th entry of its $i$th row. A matrix $M$ is said to be row stochastic if all entries are non-negative and each row sums to $1$. 

A graph is defined by $\mathcal G=(\mathcal V,\mathcal E)$, where $\mathcal V=\{1,\dots,n\}$ is the node set and $\mathcal E\subseteq \mathcal V\times\mathcal V$ is the set of edges, such that $(i,j)\in\mathcal E$ if and only if (iff) there is a direct edge from $i$ to $j$. In general, a graph is assumed to be directed, unless it holds that $(i,j)\in\mathcal E \Leftrightarrow (j,i)\in\mathcal E$, i.e., all connections are bi-directional. In this case, the network is said to be undirected. Given a real square matrix $M\in\mathbb R^{n\times n}$ with $m_{ij}\geq 0$, we can associate a weighted (directed) graph $\mathcal{G}[M]:= (\mathcal{V}, \mathcal{E}_M, M)$ to $M$, where $M$ is called the (weighted) adjacency matrix of this graph and $(i,j)\in\mathcal{E}_M$ iff $m_{ij}>0$. 

\subsection{Game Theory}
A game $\Gamma=(\gV,\gA,\vect{f})$ is defined by three terms: 1) a set of players $\mathcal{V}=\{1,\ldots,n\}$, 2) a set of strategies $\gA$, and 3) a payoff function $\vect{f}$. Player $i$ can choose their strategy from the set $\mathcal A$, which is denoted by $z_i\in \gA$. Strategies are gathered in the vector $\vect{z}=[z_1,\dots,z_n]^\top$ which is called strategy profile. We further define the vector $\vect{z_{-i}}\in \gA^{n-1}$ that represents the strategy profile of all agents except for agent $i$. 
For each player $i\in\mathcal V$, we define their payoff function $f_i:\mathcal A\times\mathcal A^{n-1}\to\mathbb R$, denoted as $f_i(z_{i},\vect{z_{-i}})$, which associates to player $i$ the payoff that $i$ would receive by choosing strategy $z_i\in\mathcal A$ when the action profile of the others is $\vect{z_{-i}}$. Payoff functions are gathered in the vector function $\vect{f}=[f_i,\dots,f_n]^\top$.
Given a game $\Gamma=(\gV,\gA,\vect{f})$, we can define for each player $i$ the notion of the best response given the strategy profile of the others, $\vect{z_{-i}}$, as the strategy that they should play in order to maximize their payoff.
\begin{definition}[Best response]\label{def:best_response}
    Given a game $\Gamma =(\mathcal{V},\gA,\bm{f})$, the best response strategies for player $i \in\mathcal{V}$ 
    is defined as
\begin{equation}\label{eq:best_response}
    \mathcal{B}_i(f_i(\cdot,\vect{z_{-i}}))\triangleq\{\zeta^{*}:\;\zeta^{*}\in\argmax\nolimits_{\zeta\in \gA}f_i(\zeta,\vect{z_{-i}})\}.
\end{equation}
\end{definition}

A NE is a key concept in game theory, representing strategy profiles $\vect z$ where no player has any incentives to deviate from the current strategy, as formalized in the following.

\begin{definition}[Nash equilibrium]\label{def_Nash_equ}
The strategy profile $\vect{z}^{*}$ is a NE iff $f_i(z_{i}^{*},\vect{z_{-i}}^{*})\geq f_i(z_{i}^{'},\vect{z_{-i}}^{*})$,  $\forall\,z_{i}^{'}\in \gA$ and $\forall\,i\in\mathcal{V}$.
\end{definition}

In other words, at a NE, all players are playing their current best response to the other players' strategies, yielding the following classical result from~\cite{fudenberg1991game}.

\begin{lemma}\label{lemma_Nash_equ}
The strategy profile $\vect{z}^{*}$ is a NE iff 
$z_{i}^{*}\in \mathcal{B}_i\big(f_i(\cdot,\vect{z_{-i}^{*}})\big)$, for all $i\in\mathcal{V}$.
\end{lemma}

A network game is a game $\Gamma=(\gV,\gA,\vect{f})$ where the payoff function is modulated by a network $\mathcal G=(\mathcal V,\mathcal E)$. Specifically, the payoff that each agent $i$ would receive depends only on the strategy of their neighbors $\mathcal N_i:=\{j:(i,j)\in\mathcal E\}$. In the following, we present two examples of network games that will be used in this paper.
\begin{example}[Network coordination and anti-coordination]
Consider a population of $n \geq 2$ players, where each player has a binary strategy set $\mathcal A=\{-1,+1\}$. Player interactions are represented via a weighted graph $\mathcal G[A]=(\mathcal V,\mathcal E_A, A)$, where the weight $a_{ij}$ represents the impact of the strategy of $j$ on $i$'s payoff. 
Player $i$ can be either a coordinating agent or an anti-coordinating agent. A coordinating agent receives a positive contribution to their payoff from each of their neighbors with whom they coordinate (i.e., they adopt the same strategy). On the contrary, an anti-coordinating agent would receive a positive contribution to their payoff from each of their neighbors with whom they do not coordinate (i.e., they adopt different strategies).

In mathematical terms, given strategy profile $\boldsymbol{x}$, each player $i\in\mathcal{V}$ is associated with the following payoff function:
\begin{equation}\label{eq:payoff_coordination}
\begin{split}
f_i\left(x_i, \boldsymbol{x_{-i}}\right)&=\frac{\varepsilon_i}{4} \sum\nolimits_{j \in \mathcal{V}} a_{i j}\left[\left(1-x_j\right)\left(1-x_i\right) \right.\\
&\left.+(1+\alpha)\left(1+x_j\right)\left(1+x_i\right) \right] 
\end{split}
\end{equation}
with 
\begin{equation}
    \varepsilon_i=\left\{\begin{array}{ll}+1&\text{if $i$ is a coordinating agent,}\\
    -1&\text{if $i$ is an anti-coordinating agent.}
\end{array}\right.
\end{equation}
The constant $\alpha \geq 0$ is a non-negative parameter that captures the relative advantage of strategy $+1$ with respect to strategy $-1$: a coordinating agent $i$ would get a unit payoff for coordinating on action $-1$ and payoff $1+\alpha$ for coordinating on action $+1$; an anti-coordinating agent would get the opposite payoffs.

Observe that the best-response strategy for the payoff function in \eqref{eq:payoff_coordination} as defined in \eqref{eq:best_response} yields a threshold such that $+1 \in \mathcal{B}_i\left(f_i(\cdot, \boldsymbol{x})\right)$ iff a discriminant quantity $\delta_i(\boldsymbol{x})$ satisfies
\begin{equation}
    \delta_i(\boldsymbol{x})=\varepsilon_i\sum\nolimits_{j \in \mathcal{V}} a_{i j}\left[2 x_j+\alpha\left(1+x_j\right)\right] \geq 0.
\end{equation}

In other words, action $+1$ is the best response for a coordinating (anti-coordinating, respectively) agent if a sufficiently large (small, respectively) fraction of the neighbors adopts it (weighted by the entries of matrix $A$ and by the relative advantage $\alpha$). 
\end{example}

\begin{remark}Network coordination and anti-coordination games have been extensively used in the context of decision-making~\cite{Jackson2015}. On the one hand, network coordination games model strategic complements, whereby a player choosing a strategy makes that strategy more appealing to other players. This captures, e.g., the adoption of social norms and conventions, or the adoption of new technological assets~\cite{montanari2010spread_innovation,young2011dynamics,ye2021nat}. On the other hand, network anti-coordination games are representative of scenarios with strategic substitutes, whereby a player choosing a strategy makes that strategy less appealing to others. This is the case, e.g., of competition for resources that can be subject to congestion~\cite{Bramoull2007,alvarez2004network}. In real-world scenarios, there could also be situations in which part of the population has a tendency to conform with others and part of the population aims at differentiation, yielding heterogeneous $\varepsilon_i$ across the population~\cite{ramazi2016networks,
Vanelli2020}.
\end{remark}

\begin{example}[Opinion dynamics]
The French-DeGroot opinion dynamics model~\cite{french1956_socialpower,degroot1974reaching} can be cast as a network game~\cite{marden2009game}. 
We consider a population of $n\geq 2$ players, where each player has a continuous strategy $\mathcal A=[-1,+1]$, representing the agent's opinion. We denote by $y_i(t)\in\mathcal A$ the opinion of $i$ at discrete time $t$. 
agents share their opinions through interactions on a weighted graph $\mathcal{G}[W]=\left(\mathcal{V}, \mathcal{E}_W, W\right)$, where $W$ is assumed to be row stochastic. The weight $w_{ij}$ represents the social influence of $j$ on $i$ in its opinion formation process. The opinion formation process is captured by a network game, with a payoff function 
\begin{equation}\label{eq:payoff_opinion}
    f_i\left(y_i, \vect{y_{-i}}\right)=-\frac{1}{2}\sum\nolimits_{j \in \mathcal{V}} w_{i j}\left(y_i-y_j\right)^2.
\end{equation}
The best-response strategy to \eqref{eq:payoff_opinion} yields 
\begin{equation}\label{eq_best_response_opinion}
  \mathcal{B}_i\left(f_i(\cdot, \vect{y})\right) =\sum\nolimits_{j \in \mathcal{V}} w_{i j} y_j(t),  
\end{equation}
i.e., a weighted average of the opinions shared by the neighbors~\cite{degroot1974reaching}. An additional term can be added to the payoff function to account for agents' attachment to their prejudices~\cite{friedkin1990_FJsocialmodel}.
\end{example}

\section{Model}\label{sec:model}

In this section, we present the coevolutionary model of actions and opinions for networks of coordinating and anti-coordinating agents. We start by formulating the payoff function, building on~\cite{aghbolagh2023coevolutionary}, and discussing its main terms. Then, we present the model dynamics by defining the update rule and expressing a closed-form expression for it.

\subsection{Payoff Function}
Consider a population of $n \geq 2$ agents. According to the coevolutionary dynamics proposed in~\cite{aghbolagh2023coevolutionary}, at each discrete-time instant $t \in\mathbb Z_{>0}$,
each agent $i\in \mathcal{V}$ is characterized by a binary state variable $z_i(t)=(x_i(t),y_i(t))\in \gA=\{-1,1\}\times [-1,1]$. The first entry, $x_i(t)\in\{-1,+1\}$, represents the agent's binary action, while the second entry, $y_i(t) \in [-1,1]$, is the agent's opinion, viz. their orientation and preference towards the two action options. Namely, $y_i(t) > 0$ and $y_i(t) < 0$ indicate the agent~$i$ prefers action $+1$ and action $-1$, respectively, with the magnitude capturing the strength of preference. Agents' state variables are gathered in a $n$-by-$2$ matrix $\vect z(t)=[\vect x(t),\vect y(t)]=[z_1(t)^{\top},z_2(t)^{\top},\ldots,z_n(t)^{\top}]^{\top}$, which represents the state of the population at time $t$. 

Agents interact on social networks, exchanging opinions on a weighted graph $\mathcal{G}[W]$ and observing the actions of others on a weighted graph $\mathcal{G}[A]$, with weighted adjacency matrices $W$ and $A$ that are both assumed to be row-stochastic. Hence, each entry in $W$ and $A$ represents the relative social influence of other players on opinions and actions, respectively.

In~\cite{aghbolagh2023coevolutionary}, agents revise their actions assuming that they engage in a game, with a payoff function that accounts for coordination with others' actions, the social influence of others' opinions, and self-consistency between their own actions and opinions. Here, we build on that model and extend it to allow agents to anti-coordinate with their peers. Specifically, given a state $\vect{z} = (\vect{x}, \vect{y})$, we define the payoff function $f_i(z_i, \vect{z_{-i}})$ for agent $i$ to adopt strategy $z_i= (x_i,y_i)$ as
\begin{equation}\label{def:payoff_func_coo_anti}
\begin{split}
    &f_i (z_i,\vect{z_{-i}}) = \frac{\varepsilon_i\lambda_i(1-\beta_i)}{4}\sum\nolimits_{j\in \mathcal V}a_{ij}\Big[(1-x_i)(1-x_j)\\
     &\qquad+(1+\alpha)(1+x_j)(1+x_i)\Big]\\
     &-\frac{1}{2}(1-\lambda_i)\beta_i\sum\nolimits_{j \in \mathcal V}w_{ij}(y_i\hspace{-.07cm}-y_j)^2  
     -\frac{1}{2}\lambda_i\beta_i(y_i-x_i)^2.
\end{split}
\end{equation}
This expression accounts for three terms. The first term involves only the action and coincides with the payoff for a coordination or anti-coordination game in \eqref{eq:payoff_coordination}. The second term depends only on the opinion and is associated with the payoff of a French-DeGroot opinion dynamics in \eqref{eq:payoff_opinion}. The third term accounts for the human tendency to behave consistently with their opinions, penalizing the adoption of actions that are distant from them. While the last two terms coincide with those in \cite{aghbolagh2023coevolutionary}, the first one differs by the inclusion of the parameter $\varepsilon_i\in\{-1,+1\}$, accounting for whether agent $i$ is a coordinating or an anti-coordinating agent. In particular, $i$ engages in a pairwise coordination game if $\varepsilon_i=+1$ or a pairwise anti-coordination game if $\varepsilon_i=-1$. The payoff in \eqref{def:payoff_func_coo_anti} coincides with the one 
 from~\cite{aghbolagh2023coevolutionary} if $\varepsilon_i=+1$ for all $i\in\mathcal V$. Parameters $\lambda_i\in[0,1]$ and $\beta_i\in[0,1]$ weight the contribution of actions and opinions in determining agent $i$'s payoff, respectively. These two parameters also regulate the coupling of opinion formation and decision-making in the third term in \eqref{def:payoff_func_coo_anti}.  
    Similarly to~\cite{aghbolagh2023coevolutionary}, a fourth term can be added to account for agents' prejudices. Here, we omit this term, since it is not relevant to our analyses.

\subsection{Discussion on the model}\label{sec:discussion}

The game with payoff function in \eqref{def:payoff_func_coo_anti} was originally proposed in \cite{aghbolagh2023coevolutionary} with $\varepsilon_i=1$ for all $i\in\mathcal V$ as a combination of a network coordination game (first term in the payoff structure) and a French-DeGroot opinion dynamics~\cite{friedkin1990_FJsocialmodel}, which are coupled by the consistency term. Such a payoff function was proposed to capture the complex intertwining between opinion formation processes and decision-making in social networks. In fact, empirical evidence suggests that the exchange of opinions can affect agents' actions, while the way people form opinions can be influenced by observing the actions of others~\cite{henry2003beyond,haidt2001emotional,gavrilets2017collective}. However, the model was proposed in the context of actions with strategic complements, whereby the decision-making is captured by a coordination mechanism, and both opinion sharing and observation of others' actions yield a positive feedback loop.

Here, we relax this assumption, and we allow to have strategic substitutes, and thus, anti-coordinating agents. In fact, for decisions concerning actions with strategic substitutes, opinion formation, and decision-making are also deeply intertwined. For instance, when a similar service can be offered by different providers who are subject to congestion (e.g, two different restaurants to have lunch at or two different streets to take in an infrastructure network), agents still share their opinions on which provider they believe is better/preferred, and this still has a (positive) influence on others' opinions, even though the observation of others' actions yields now a different feedback effect: the provider that is less congested becomes more attractive, yielding an anti-coordination game. 

In general, complex decision-making may involve, to some extent, both strategic complements and substitutes, and agents can be heterogeneously more subject to one or the other. This is the case, e.g., of marketing decisions, where there are agents who prefer to conform with the majority, while others prefer to differentiate from it. For this reason, each agent $i\in\mathcal V$ is characterized by a parameter $\varepsilon_i\in\{-1,+1\}$, which determines the mechanism (anti-coordination or coordination) that affects the agent. 
Finally, it is worth noticing that opinions always provide a positive social influence, consistent with empirical observations on the limited impact of negative social influence~\cite{Takcs2016}.
\subsection{Update Rule}

Agents revise their strategy by maximizing their payoff function jointly with respect to their two decision variables (actions and opinions), as proposed in~\cite{aghbolagh2023coevolutionary}. 
At each time $t$, each agent is either activated or inactive. An agent $i$ that is activated updates their opinion and action simultaneously according to the best-response update rule:
\begin{equation}\label{eq:async_update}
 z_i(t+1) = \mathcal{B}_i(f_i(\cdot,\vect{z_{-i}}(t))).
\end{equation}
where $\mathcal{B}_i$ is the best-response strategy defined in \eqref{eq:best_response} and
$f_i(\cdot,\vect{z_{-i}}(t))$ is the payoff from \eqref{def:payoff_func_coo_anti}, with the convention that, if $\mathcal{B}_i(f_i(\cdot,\vect{z_{-i}}(t)))$ comprises multiple elements, then it is chosen the one with $x_i(t+1)=x_i(t)$, which is unique for $\beta_i>0$ (see~\cite{aghbolagh2023coevolutionary}), to account for social inertia in human decision-making~\cite{ye2021nat}. For all other agents who are inactive, their strategies are unchanged: $z_k(t+1) = z_k(t)$. In other words, active agents tend to maximize their current payoff, based on the current action of others, as supported by empirical evidence~\cite{mas2016behavioral}. 
Following~\cite{aghbolagh2023coevolutionary}, we derive a closed-form expression for the update rule of the coevolutionary dynamics by explicitly writing the maximizers of~\eqref{def:payoff_func_coo_anti}.
\begin{proposition}\label{prop:dyn}
Given a state $\vect{z}=(\vect{x},\vect{y})\in \gA$ and an agent $i\in\mathcal{V},$ let us define the discriminant 
\begin{equation}\label{delta_i(z)}
 \begin{aligned}
  \delta_i(\vect{z}(t)) & = \frac{\varepsilon_i\lambda_i(1-\beta_i)}{2}\sum_{j\in \mathcal V}a_{ij}\left[2x_j(t)+\alpha(1+x_j(t))\right]\\
   &+2(1-\lambda_i)\lambda_i\beta_i\sum\nolimits_{j \in \mathcal V}w_{ij}y_j(t).
 \end{aligned} 
\end{equation} 
If $i\in\mathcal V$ is activated at time $t$, the state is updated to
\begin{subequations}\label{update_rule_coo_xy}
\begin{align}
    x_i(t+1)&=\mathcal{S}(\delta_i(\vect{z}(t)))\label{action_update}\\ 
    y_i(t+1)&=(1-\lambda_i)\sum\nolimits_{j\in\mathcal{V}}\hspace{-.2cm}w_{ij}y_j(t)\label{opinion_update}
            +\lambda_i\mathcal{S}(\delta_i(\vect{z}(t))),
\end{align}
\end{subequations}
where 
\begin{equation}
     \mathcal{S}(\delta_i(\vect{z}(t)))=
     \begin{cases}
     +1 & \text{if} \ \delta_i(\vect{z}(t))>0,\\    
     x_i(t) & \text{if} \ \delta_i(\vect{z}(t))=0,\\
     -1 & \text{if} \ \delta_i(\vect{z}(t))<0.
     \end{cases}
\end{equation}
\end{proposition}
\begin{proof}
    Case $\varepsilon=+1$ is proved in \cite{aghbolagh2023coevolutionary}. The proof for $\varepsilon=-1$ follows the same arguments, and is thus omitted.
\end{proof}
Next, we formalize two standing assumptions that will be enforced for the rest of the paper. First, we assume that the population is homogeneous, second, we pose some mild constraints on the activation sequence, which are quite standard in the related literature~\cite{xiao2006state,zino2020two}.

\begin{assumption}[Homogeneous population]\label{as:uniform}
The parameters are homogeneous, i.e., for all $i\in\mathcal V$, there holds  $\lambda_i=\lambda \in (0,1),\beta_i=\beta\in(0,1)$, and $\varepsilon_i=\varepsilon\in\{-1,+1\}$. Moreover, 
the two graphs are undirected, i.e., $W=W^{\top}$ and $A=A^{\top}$.
\end{assumption}

\begin{assumption}[Activation sequence]\label{as:activation}
There exists a finite positive integer $T$ such that, for any $t \geq 0$, every agent activates at least once in the interval $t,t+1,\ldots,t+T-1$.
\end{assumption}

Under Assumption~\ref{as:uniform}, all agents have the same $\varepsilon$. Hence, we can identify two scenarios of interest: i) Networks of \emph{coordinating agents}, where $\varepsilon=+1$; and ii) Networks of \emph{anti-coordinating agents}, where $\varepsilon=-1$. In the rest of this paper, we present results on these two scenarios. First, in Section~\ref{sec: results_coo}, we consider a population of coordinating agents, where \cite{aghbolagh2023coevolutionary} guarantees convergence to an equilibrium, and we focus on characterizing conditions under which the system converges to a consensus. Second, in Section~\ref{sec: results_anti_coord}, we consider a population of anti-coordinating agents and, after proving convergence to an equilibrium, we investigate consensus and polarized equilibria and their basins of attraction.

\section{Results for Coordinating Agents}\label{sec: results_coo}
In the presence of only coordinating agents $(\varepsilon_i=+1,\,\forall\,i\in\mathcal V)$, the  results from~\cite{aghbolagh2023coevolutionary} guarantee that the coevolutionary dynamics under Assumptions~\ref{as:uniform}--\ref{as:activation} converges to a NE. However, in general, multiple equilibria may be present. In~\cite{aghbolagh2023coevolutionary}, some conditions have been established under which polarized equilibria exist and are stable. Here, instead, we focus on consensus equilibria, which are defined as follows.

\begin{definition}[Consensus configuration]\label{def_NE_consensus}
   We say that the vector $\vect{z}(t)=[\vect{x}(t), \vect{y}(t)]$ is a  consensus configuration at time $t$ if $\vect{x}(t)= \vect{y}(t)=\vect{1}_n$ or $\vect{x}(t)=\vect{y}(t)=-\vect{1}_n.$ 
\end{definition}

We now prove that consensus configurations are always NEs of the game when all players are coordinating agents.  

\begin{theorem}\label{thm:NE_consesus}
Consider the coevolutionary model in \eqref{update_rule_coo_xy}.
Suppose that Assumption~\ref{as:uniform}  holds and $\varepsilon=1$.
Then the consensus configurations $\vect{z}^{*}= [\vect{1}_n,\vect{1}_n]$ and $[-\vect{1}_{n},-\vect{1}_{n}]$
are NE.
\end{theorem}
\begin{proof}The proof is reported in Appendix~\ref{app:thm:NE_consesus}. 
\end{proof}

In the following, we establish conditions under which the dynamics converge to such a consensus configuration.

\subsection{Basin of Attraction of Consensus Configuration}
Given a time instant $t\in \mathbb Z_+$, and $\vect z(t) = (\vect x(t), \vect y(t))$, we define the following mutually exclusive sets
\begin{align*}
    \mathcal{V}_{1}^{+}(t) & :=\{j:\sgn(y_j(t))\geq 0, x_j(t)=+1\}, \\
    \mathcal{V}_{-1}^{+}(t)& :=\{j:\sgn(y_j(t))>0, x_j(t)=-1\}, \\
    \mathcal{V}_{1}^{-}(t)& :=\{j:\sgn(y_j(t))<0, x_j(t)=+1\},\\
    \mathcal{V}_{-1}^{-}(t)& :=\{j:\sgn(y_j(t))\leq 0, x_j(t)=-1\}.
\end{align*}
with $\sgn(0) = 0$.
Intuitively, $\mathcal{V}_{1}^{+}(t)$ and $\mathcal{V}_{-1}^{-}(t)$ are the sets of agents whose actions and opinions are aligned (the agent's opinion indicates a preference for their current action). On the contrary, $\mathcal{V}_{-1}^{+}(t)$ and $\mathcal{V}_{1}^{-}(t)$ are the sets of agents who are misaligned, i.e., the sign of their opinion differs from their action. Before presenting our results, we report two graph-theoretic definitions.

\begin{definition}[Cohesive set]\label{def:cohesive set}
Let $\mathcal{G}[M]=(\mathcal{V},\mathcal{E}_M,M)$ be a graph with associated adjacency matrix $M$. Then, for a given constant $q\in (0,1)$, a set $S\subset$ $\mathcal{V}$ is called $q$-cohesive if 
\begin{equation}
\frac{\sum_{j \in S} m_{ij}}{\sum_{j \in \mathcal{V}}m_{ij}}\geq q,\;\forall\,i \in S.
\end{equation}
Conversely, a set $S\subset$ $\mathcal{V}$ is called $q$-diffusive if 
\begin{equation}
\frac{\sum_{j \in S} m_{ij}}{\sum_{j \in \mathcal{V}}m_{ij}}<q,\;\forall\,i \in S.
\end{equation}
\end{definition}

\begin{theorem}\label{thm:basins of attraction_vp_vnp}
Consider the coevolutionary model in \eqref{update_rule_coo_xy}.
Suppose that Assumptions~\ref{as:uniform}--\ref{as:activation} hold and $\varepsilon=1$, and that the initial condition satisfies $\vect{z}(0)\in\{-1,1\}^n\times(0,+1]^n$. If the set $\mathcal{V}_{1}^{+}(0)$ is $\frac{1}{\alpha+2}$-cohesive and $\mathcal{V}_{-1}^{+}(0)$ is $\frac{\alpha+1}{\alpha+2}$-diffusive, 
then $\lim_{t\to\infty}\vect{z}(t)=[\vect{1}_n,\vect{1}_n]$.
\end{theorem}
\begin{proof}
The proof is reported in Appendix~\ref{app:thm:basins of attraction_vp_vnp}.
\end{proof}

\begin{remark}
When $\alpha =0$, the condition reduces to requiring $\mathcal{V}_{1}^{+}(0)$ to be $\frac{1}{2}$-cohesive, and the complementary set to be $\frac{1}{2}$-diffusive. 
On the opposite limit, $\alpha \rightarrow \infty$, the condition on $\mathcal V_1^+(0)$ intuitively reduces to requiring non-zero cohesiveness. 
\end{remark}

We introduce the following assumptions to facilitate our analysis of some (but not all) of the results in our paper.
\begin{assumption}\label{as:polarization}
We assume that  $A=W$ and $\alpha=0.$
\end{assumption}

Under Assumption~\ref{as:polarization}, the action and opinion update rules are simplified, because \eqref{delta_i(z)} reduces to
\begin{equation}\label{eq:delta_ass1_4_coo}
    \delta_i(\vect{z(t)})=2\lambda\beta(1-\lambda)\sum_{j\in \gV}w_{ij}y_j(t)+(1-\beta)\lambda\sum_{j\in\gV} w_{ij}x_j(t).
\end{equation}

We now explore the conditions that lead two initially polarized groups to reach a consensus equilibrium. To begin, we introduce the following definition.

\begin{definition}[Polarized configuration]\label{def_polar}
    We say that a vector $\vect{z}(t)$ is a polarized configuration at time $t$ if $\mathcal{V}_{1}^{+}(t) \cup \mathcal{V}_{-1}^{-}(t) = \mathcal{V}$, and both $\mathcal{V}_{1}^{+}(t) \neq \emptyset$ and $\mathcal{V}_{-1}^{-}(t) \neq \emptyset$. Specifically, we say that $\vect{z}(t)$ is polarized with respect to the partition $(\mathcal{V}_{1}^{+}(t),\mathcal{V}_{-1}^{-}(t))$. 
\end{definition}
In other words, in polarized configurations, the opinions and actions of every agent are aligned, but they are not at a consensus. Note that in \cite{aghbolagh2023coevolutionary}, a result is established for when a polarized initial state leads to convergence to a polarized equilibrium, whereas here, we identify conditions for when instead there is convergence to a consensus equilibrium.

\begin{figure}
 \centering
 \subfloat[\label{fig_heatmap_energy}]{\includegraphics[width=\linewidth]{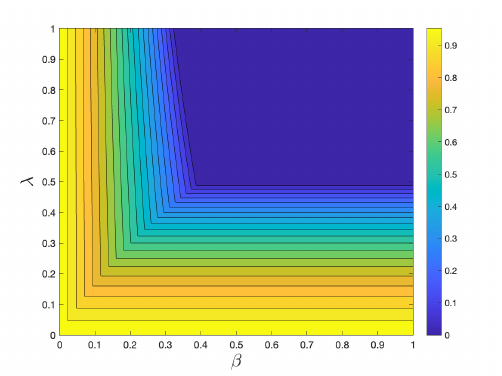}}\\
 \vspace{-0.3cm}
 \subfloat[\label{fig_heatmap_ins}]{\includegraphics[width=\linewidth]{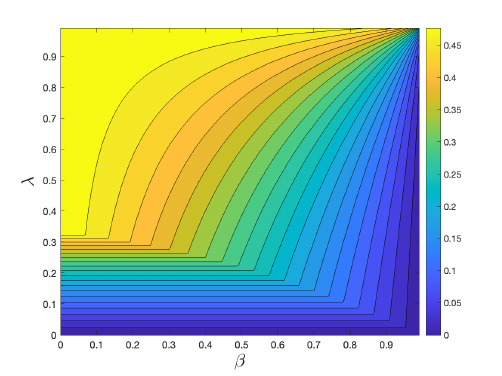}} 
  \caption{Contour plots of the conditions in \eqref{eq:condition_consensus_VnVp_coo_thm} from Theorem~\ref{thm:basins of attraction_vnvp_consensus}. Panel (a): $f(\lambda, \beta)$ corresponds to the right-hand-side (RHS) of~\eqref{eq:condition_consensus_VnVp_coo_thm_vp}.
  Panel (b): $f(\lambda, \beta)$ corresponds to the RHS of~\eqref{eq:condition_consensus_VnVp_coo_thm_vn}.}
 \label{fig:thm3_contour} 
\end{figure}

\begin{theorem}\label{thm:basins of attraction_vnvp_consensus}
Consider the coevolutionary model in \eqref{update_rule_coo_xy}.
Suppose that Assumptions~\ref{as:uniform}--\ref{as:polarization} hold and $\varepsilon=1$. Suppose $\vect{z}(0)$ is polarized, with respect to a partition $(\mathcal{V}_{1}^{+}(0),\mathcal{V}_{-1}^{-}(0))=:(\mathcal{V}_p,\mathcal{V}_n)$. If the below inequalities are satisfied,
\begin{subequations}\label{eq:condition_consensus_VnVp_coo_thm}
\begin{align}
    \sum\nolimits_{j\in\mathcal{V}_{1}^{+}(0)}w_{ij}&>\max\left\{\frac{1-2\lambda}{1-\lambda},1-\frac{1-\beta}{2(1-\lambda\beta)}\right\},  i\in\mathcal{V}_p\label{eq:condition_consensus_VnVp_coo_thm_vp},\\
\sum\nolimits_{j\in\mathcal{V}_{-1}^{-}(0)}\hspace{-.1cm}w_{ij}\hspace{-.07cm}&< \hspace{-.07cm}\min\left\{\frac{\lambda}{1-\lambda}, \frac{1-\beta}{2(1-\lambda\beta)}\right\},\,   i\in\mathcal{V}_n,\label{eq:condition_consensus_VnVp_coo_thm_vn}
\end{align}
\end{subequations}
Then, it holds $\lim_{t\to\infty}\vect{z}(t)=[\vect{1}_n,\vect{1}_n]$.
\end{theorem}

\begin{proof}
The proof is reported in Appendix~\ref{app:thm:basins of attraction_vnvp_consensus}.     
\end{proof}

\begin{remark}
If \eqref{eq:condition_consensus_VnVp_coo_thm_vp} is satisfied, then $\mathcal{V}_p$ is $\frac{1}{2}-$cohesive
since the quantities $\frac{1-2\lambda}{1-\lambda}$ and $1-\frac{1-\beta}{2(1-\lambda\beta)}$ are always greater than $\frac{1}{2}.$
\end{remark}

We now present contour plots to explain the inequalities in \eqref{eq:condition_consensus_VnVp_coo_thm} more intuitively.
In this paper, we make repeated use of similar contour plots, and thus we first clarify the horizontal and vertical axes, as well as the colors used in these contour plots, which will remain consistent across all sections. The 2D contour plots help to visualize the relationship between $\beta$ and $\lambda$ (contributions in the decision-making model) and a function $f(\lambda,\beta)$ that arises for various inequalities identified in theorems. The X-axis plots $\beta$, the parameter that represents the contribution of the opinion in determining an agent's payoff, while the Y-axis plots $\lambda$, which represents the contribution of the action in determining an agent's payoff. The contour colors range from blue to orange, corresponding to $f(\lambda,\beta)$ taking the smallest to largest value, respectively. Fig.~\ref{fig:thm3_contour} presents the two inequalities in \eqref{eq:condition_consensus_VnVp_coo_thm}. From these two contour plots, it appears that $f(\lambda, \beta)$ is minimized in Fig.~\ref{fig_heatmap_energy} and maximized in Fig.~\ref{fig_heatmap_ins} for high $\beta$ values, and medium to high $\lambda$ values. Put another way, given $W$ and $\vect z(0)$, the inequalities in \eqref{eq:condition_consensus_VnVp_coo_thm} are more likely to be satisfied with high $\beta$, and medium to high $\lambda$. This suggests that one can reach a consensus configuration equilibrium, from an initially polarized state, if the opinion dynamics strongly impact the coevolutionary process (and have a stronger impact on coordination). 

\begin{example}\label{ex:coo_consensus}
Consider a network of $n=30$ agents partitioned into two groups, $\mathcal{V}_p= \{1,\ldots,15\}$ and $\mathcal{V}_n=\{16,\ldots,30\}$. Given $\vect z(0)$, we set weights uniformly at random and then re-scale them so that conditions in \eqref{eq:condition_consensus_VnVp_coo_thm} are met with $\mathcal{V}_p=\mathcal{V}_{1}^{+}(0)$ and $\mathcal{V}_n = \mathcal{V}_{-1}^{-}(0)$, and weights are row-stochastic.
The network obtained is illustrated in Fig.~\ref{fig:coo_consensus}(a). The parameters are $\lambda=0.8,\beta=0.6$.  As shown in Fig.~\ref{fig:coo_consensus}(b) and ~\ref{fig:coo_consensus}(c), the final actions and opinions converge to a consensus.
\end{example}

\begin{figure}
  \centering
  \includegraphics[width=\linewidth]{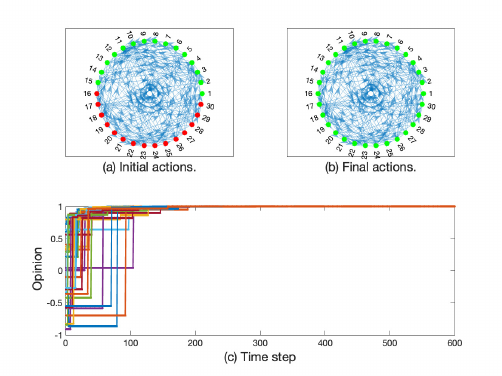}
  \caption{Network and actions at (a) $t=0$ and (b) $t=600$. Nodes are in green if $\sgn(y_i)=x_i=+1$ and in red if $\sgn(y_i)=x_i=-1$. In (c), the temporal evolution of the opinion is shown to converge to consensus.}
  \label{fig:coo_consensus}
\end{figure}
\section{Results for Anti-coordinating Agents}\label{sec: results_anti_coord}


\subsection{Convergence}

We start the analysis of the coevolutionary model with anti-coordinating agents $(\varepsilon_i=-1,\,\forall i\in\mathcal V)$ by proving the following convergence result.

\begin{theorem}\label{thm:anti_convergence_all}
Consider the coevolutionary model in \eqref{update_rule_coo_xy}.
Suppose that Assumptions~\ref{as:uniform} and \ref{as:activation} hold and $\varepsilon=-1$. Then, the state of the system $\vect{z}(t)$ converges to a steady state $\vect{z}^* = (\vect x^*, \vect y^*)$. Specifically, the action vector $\vect{x}(t)$ converges to $\vect{x^*}$ in finite time, while the opinion vector $\vect{y}(t)$ converges to $\vect{y^*}$ asymptotically.
\end{theorem}

\begin{proof}
The proof is reported in Appendix~\ref{app:thm:anti_convergence_all}.  
\end{proof}


After having established convergence in the presence of anti-coordinating agents, we start analyzing some equilibria of interest, namely, consensus and polarized configurations.

\subsection{Consensus Equilibria}
We start by analyzing whether a consensus equilibrium exists,  establishing the following result.
\begin{theorem}\label{thm:consensus_ equilibrium_anti} Consider the coevolutionary model in \eqref{update_rule_coo_xy}.
Suppose that Assumptions~\ref{as:uniform}--\ref{as:activation} hold and $\varepsilon=-1$. Then, $\vect{z}^*=[\vect{1}_n,\vect{1}_n]$  
is an equilibrium of the coevolutionary model on the two-layer network $\gG$ iff, for any $i\in\mathcal{V}$, it holds  
\begin{equation}\label{eq:condition_A}\sum\nolimits_{j\in\mathcal{V}}a_{ij}<\frac{2(1-\lambda)\beta}{(1-\beta)(1+\alpha)}.
\end{equation}
\end{theorem}

\begin{proof}The proof is reported in Appendix~\ref{app:consensus_ equilibrium_anti}.
\end{proof}

\begin{figure}
  \centering
  \includegraphics[width=\linewidth]{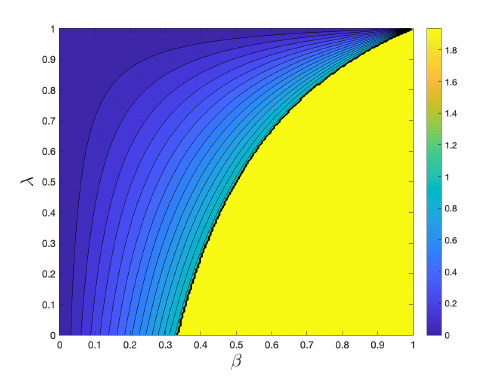}
\caption{Contour plot of the RHS in \eqref{eq:condition_A} in Theorem~\ref{thm:consensus_ equilibrium_anti} when $\alpha=0$. The black regions in the plot represent areas where the contour lines are denser, indicating faster changes in values.}
  \label{fig: 2D_anti_coo_contour}
\end{figure}

We give a 2D contour plot of the RHS of the inequality in Theorem~\ref{thm:consensus_ equilibrium_anti} to observe whether a given influence network $\mathcal{G}[A]$ easily satisfies the condition when $\lambda$ and $\beta$ take different values with $\alpha=0$. 
 Fig.~\ref{fig: 2D_anti_coo_contour} depicts the condition in \eqref{eq:condition_A} with $\alpha = 0$. From the contour plot, it appears that $\frac{2(1-\lambda)\beta}{(1-\beta)}$ is maximized for high $\beta$ values, and low to medium $\lambda$ values. In other words, the condition in \eqref{eq:condition_A} is more likely to be satisfied when $\beta$ is large and $\lambda$ is small. Intuitively, this suggests that it is easier to reach a consensus equilibrium when the decision-making process is more affected by the opinion exchange. We also see a clear distinction between  coordinating agents and anti-coordinating agents. For the former, there is always a consensus equilibrium~\cite[Proposition~2]{raineri2025}. For the latter, the existence of such equilibrium depends on the parameters. 
In what follows, we present an example that satisfies the conditions of Theorem~\ref{thm:anti_convergence_all}.

\begin{example}\label{ex:anti_consensus}
Consider a network of $n=30$ agents with weights sampled uniformly at random and then re-scale so that the matrix is row-stochastic and so that conditions in \eqref{eq:condition_A} are met with $\alpha=0$. The network obtained is illustrated in Fig.~\ref{fig:anti_coo_ER}(a). The parameters are $\lambda=0.5,\beta=0.8$. As shown in Fig.~\ref{fig:anti_coo_ER}(b) and ~\ref{fig:anti_coo_ER}(c), the final actions and opinions converge to a consensus.
\end{example}
 
\begin{figure}
  \centering
  \includegraphics[width=\linewidth]{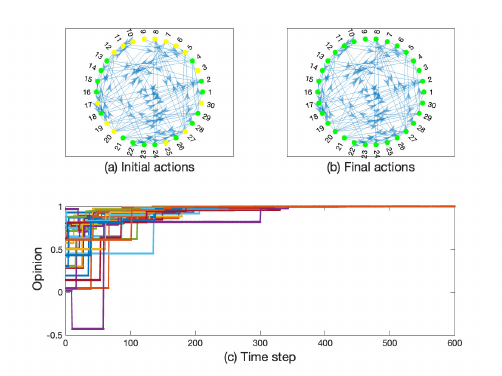}
\caption{Network and actions at (a) $t=0$ and (b) $t=600$. Nodes are in green if $\sgn(y_i)=x_i=+1$ and in yellow if $\sgn(y_i)\neq x_i$. In (c), the temporal evolution of the opinion shows convergence to consensus.}
  \label{fig:anti_coo_ER}
\end{figure}

\subsection{Polarized Equilibria}

Here, we discuss how our model may lead to the emergence of stable polarized configurations (see Definition~\ref{def_polar}). 
In particular, the following result provides a sufficient condition for the existence of a polarized equilibrium in the presence of anti-coordinating agents, and complements a similar result for coordinating agents in \cite[Theorem~2]{aghbolagh2023coevolutionary}.
\begin{theorem}\label{thm:polarized_anti}
Consider the coevolutionary model in \eqref{update_rule_coo_xy}.
Suppose that Assumptions~\ref{as:uniform}--\ref{as:polarization} hold and $\varepsilon=-1$. Suppose there exists a partitioning of the agents $(\mathcal{V}_p,\mathcal{V}_n)$ with 
\begin{equation}\label{eq:dp_dn_anti}
 \begin{aligned}
   \underline{d}_p:=\min_{i\in\mathcal{V}_p} \sum\nolimits_{j\in\mathcal{V}_p} w_{ij},  
   \bar{d}_p:=\max_{i\in\mathcal{V}_p} \sum\nolimits_{j\in\mathcal{V}_p} w_{ij},\\
   \underline{d}_n:=\min_{i\in\mathcal{V}_n } \sum\nolimits_{j\in\mathcal{V}_n }w_{ij}, 
  \bar{d}_n:=\max_{i\in\mathcal{V}_n } \sum\nolimits_{j\in\mathcal{V}_n }w_{ij}, 
\end{aligned}
\end{equation}
such that the following three conditions are satisfied:
\begin{subequations}
\begin{equation}\label{eq:pol_y_cond_anti}
   \lambda>\max \Big(\frac{\underline{d}_p-\bar{d}_n }{1+\underline{d}_p-\underline{d}_n}, \frac{\bar d_n-\underline d_p}{1+\bar d_n-\underline d_p}\Big),
\end{equation}
\begin{equation}\label{eq:pol_xp_cond_anti}
    \frac{\lambda (\underline{d}_p +\bar{d}_n - 1) + \underline{d}_p -\bar{d}_n}{1-(1-\lambda)(\bar{d}_p +\bar{d}_n-1)}+\frac{(1-\beta)(1-2\bar{d}_p)}{2 \beta  (1-\lambda )}>0,
\end{equation}
\begin{equation}\label{eq:pol_xn_cond_anti}
 \frac{\lambda (\underline{d}_n + \bar{d}_p - 1) + \underline{d}_n - \bar{d}_p}{1-(1-\lambda)(\bar{d}_n+\bar{d}_p -1)}-\frac{(1-\beta)(1-2\bar{d}_n)}{2 \beta  (1-\lambda )}>0.
\end{equation}\end{subequations}
Then, the coevolutionary model in \eqref{update_rule_coo_xy} has a polarized equilibrium with respect to the partition $(\mathcal{V}_p,\mathcal{V}_n)$.
\end{theorem}
\begin{proof}
The proof is reported in Appendix~\ref{app:thm:polarized_anti}.    
\end{proof}

Theorem~\ref{thm:polarized_anti} establishes three conditions for the existence of a polarized equilibrium. Although the last two conditions may not be intuitive, all three are easy to check. Below, we study the stability of such equilibrium, establishing a more conservative sufficient condition.

\begin{theorem}\label{thm:polarized_convergence_anti}
Consider the coevolutionary model in \eqref{update_rule_coo_xy}. Suppose that Assumptions~\ref{as:uniform}--\ref{as:polarization} hold, $\beta<1/(2-\lambda)$, and $\varepsilon=-1$. 
Suppose $\vect{z}(0)$ is a polarized configuration with respect to a partition $(\mathcal{V}_{1}^{+}(0),\mathcal{V}_{-1}^{-}(0))=:(\mathcal{V}_p,\mathcal{V}_n)$.  
If
\begin{equation}\label{eq:condition_polarized_anti}
    \frac{1-2\lambda}{1-\lambda}<\sum\nolimits_{j\in  \mathcal S}w_{ij}<\frac12\Big(1+\frac{\beta(1-\lambda)}{2\beta-\beta\lambda-1}\Big), \forall\, i\in \mathcal S,
\end{equation}
holds separately for $\mathcal S=\mathcal V_p$ and $\mathcal S=\mathcal V_n$, then $\vect{z}(t)$ is polarized for any $t$ with respect to the same partition $(\mathcal{V}_p,\mathcal{V}_n)$. Moreover, $\lim_{t\to\infty} \vect{z}(t) = \vect{z}^*$, where $\vect z^*$ is a polarized equilibrium with respect to the same partition $(\mathcal{V}_p,\mathcal{V}_n)$.
\end{theorem}
\begin{proof}The proof is reported in Appendix~\ref{app:polarized_convergence_anti}. 
\end{proof}

\begin{remark}
In~\cite[Theorem 3]{aghbolagh2023coevolutionary} a sufficient condition for existence and convergence to a polarized equilibrium for networks of coordinating agents is reported, paralleling Theorem~\ref{thm:polarized_convergence_anti}. The sufficient condition reported in \cite[Theorem 3]{aghbolagh2023coevolutionary} can be expressed in the notation of this paper as 
\begin{equation}\label{eq:condition_polarized_coo}
    \sum\nolimits_{j\in  \mathcal S}w_{ij}> \max\Big\{\frac{1-2\lambda}{1-\lambda},\frac12\Big(1+\frac{\beta(1-\lambda)}{1-\beta\lambda}\Big)\Big\}, \forall\,i\in \mathcal  S,
\end{equation}
for, separately, $ \mathcal S = \mathcal{V}_p$ and $\mathcal S = \mathcal{V}_n$. The set $S$ in \eqref{eq:condition_polarized_coo} are all $\frac{1}{2}$-cohesive sets since $\max\{\frac{1-2\lambda}{1-\lambda},\frac12(1+\frac{\beta(1-\lambda)}{1-\beta\lambda})\}\geq \frac{1}{2}$.
\end{remark}

\begin{figure}
  \centering
  \includegraphics[width=\linewidth]{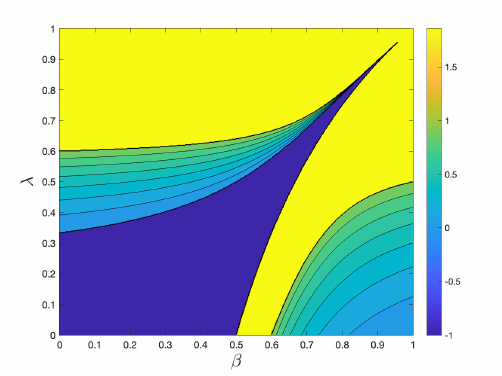}
\caption{Contour plot of the difference between the RHS and left-hand-side(LHS) in the condition of \eqref{eq:condition_polarized_anti} from Theorem~\ref{thm:polarized_convergence_anti} for agents in $\mathcal{V}_p$ or $\mathcal{V}_n$.}
  \label{fig:2D_anti_polarize_contour_diff}
\end{figure}

\begin{figure}
 \centering
 \includegraphics[width=\linewidth]{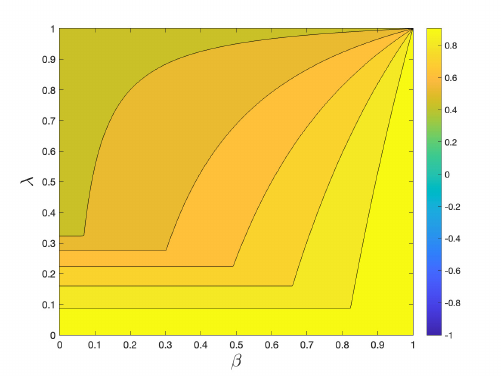}
\caption{Contour plot of the RHS in \eqref{eq:condition_polarized_coo} for agents in $\mathcal{V}_p$ or $\mathcal{V}_n$.}
  \label{fig:2D_coo_polarize_contour}
\end{figure}

We now present 2D contour plots to explain the inequalities in \eqref{eq:condition_polarized_anti} and \eqref{eq:condition_polarized_coo} more intuitively.
For the anti-coordination game, Fig.~\ref{fig:2D_anti_polarize_contour_diff} presents the two inequalities in \eqref{eq:condition_polarized_anti}. From the contour plot, it appears that $f(\lambda,\beta)= \frac12(1+\frac{\beta(1-\lambda)}{2\beta-\beta\lambda-1})-\frac{1-2\lambda}{1-\lambda}
$ is maximized in Fig.~\ref{fig:2D_anti_polarize_contour_diff} for low to medium $\beta$ values, and high $\lambda$ values. Specifically when $\lambda$ is less than about $0.3$ and $\beta$ is less than $0.5$, no network can be found to satisfy the conditions~\eqref{eq:condition_polarized_anti}. Put another way, given $W$ and $\vect z(0)$, the inequalities in \eqref{eq:condition_polarized_anti} are more likely to be satisfied with low to medium $\beta$, and high $\lambda$. This suggests that an initially polarized state could still converge to a polarized state if the anti-coordinate action dynamics strongly impact the coevolutionary process (and have a stronger impact in the anti-coordination games component). For the coordination game, Fig.~\ref{fig:2D_coo_polarize_contour} presents the two inequalities in \eqref{eq:condition_polarized_coo}. From the contour plot, it appears that $f(\lambda,\beta)= \max\{\frac{1-2\lambda}{1-\lambda},\frac12(1+\frac{\beta(1-\lambda)}{1-\beta\lambda})\}>0.5$ regardless of the values of $\lambda$ and $\beta$, which indicates that closer connections are required within the set $S$. We can thus see how different mechanisms (coordinating vs anti-coordinating) shape the existence and stability of polarized equilibria.


\begin{example}\label{ex:polarized_anti}
Consider a network of $n=30$ agents with weights sampled at random and re-scaled so that, given $\vect z(0)$,  
the graph is a balanced complete bipartite graph with $\mathcal{V}_p=\mathcal{V}_{1}^{+}(0)=\{1,\ldots,15\}$ and $\mathcal{V}_n = \mathcal{V}_{-1}^{-}(0)=\{16,\ldots,30\}$. 
Specifically, There are no edges within sets $\mathcal{V}_{1}^{+}(0)$ or $\mathcal{V}_{-1}^{-}(0)$, but there are edges between two sets $\mathcal{V}_{1}^{+}(0)$ and $\mathcal{V}_{-1}^{-}(0)$. 
It is easy to prove that the conditions in \eqref{eq:condition_polarized_coo} are met. The network obtained is illustrated in Fig.~\ref{fig:polarized_anti}(a). We set$\lambda=0.7,\beta=0.6$, which  satisfy \eqref{eq:condition_polarized_anti}. As shown in Fig.~\ref{fig:polarized_anti}(b) and \ref{fig:polarized_anti}(c), the final actions and opinions converge to a bipartite consensus, consistent with  Theorem~\ref{thm:polarized_convergence_anti}. In this case, the NE reached is not only polarized but also a bipartite consensus as depicted in Fig.\ref{fig:polarized_anti}(b) and \ref{fig:polarized_anti}(c). Notice that the above scenario is quite extreme and it is presented for the sake of having a simple illustrative example, while the conditions in \eqref{eq:condition_polarized_anti} can be satisfied also by less extreme network structures.
\end{example} 


\section{Conclusion}\label{sec:conclusion}
Building on the coordination game framework proposed by~\cite{aghbolagh2023coevolutionary}, we examined consensus equilibria and their domains of attraction. Then, we further extended the modeling framework and analysis to incorporate anti-coordinating agents, establishing a general convergence result, and identifying conditions for consensus and polarized equilibria to exist, as well as their regions of attraction. Finally, we compared our results on the configuration of equilibria and their basins of attraction in the presence od coordinating and anti-coordinating agents. Future research directions include i) extending the analysis to a complete characterization of all the equilibria of the dynamics, ii) encapsulating trust-mistrust relationships, e.g., by means of signed networks; and iii) validating the model using real-world datasets.

\begin{figure}[!htbp]
  \centering
  \includegraphics[width=\linewidth]{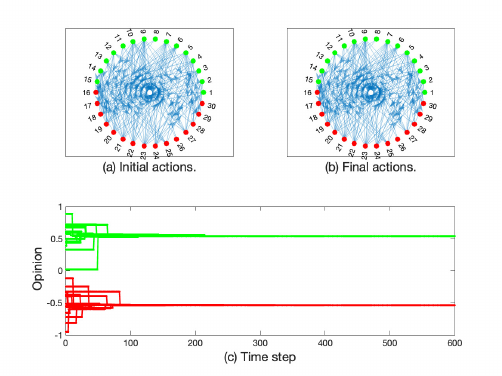}
\caption{Network and actions at (a) $t=0$ and (b) $t=600$. Nodes are in green if $\sgn(y_i)=x_i=+1$ and in red if $\sgn(y_i)=x_i=-1$. In (c), the temporal evolution of the opinion shows convergence to polarization.}
  \label{fig:polarized_anti}
\end{figure}

\appendices
\section{Proof of Theorem~\ref{thm:NE_consesus}}\label{app:thm:NE_consesus}
According to the Definition~\ref{def:best_response} and Lemma~\ref{lemma_Nash_equ}, we need to show that when $\vect{z_{-i}}=[\vect{1}_{n-1},\vect{1}_{n-1}], \forall i \in \mathcal{V}, z_i=[1,1]$ is the best response. That is, $[1,1]\in \mathcal{B}_i(f_i(\cdot,[\vect{1}_{n-1},\vect{1}_{n-1}]))$.
By Assumption~\ref{as:uniform}, we write 
\begin{equation*}\label{eq:th1_eq1}
\begin{split}  & \mathcal{B}_i (f_i(\cdot,[\vect{1}_{n-1},\vect{1}_{n-1}]))=\argmax\nolimits_{z_i} f_i(z_i,[\vect{1}_{n-1},\vect{1}_{n-1}]),  \\
&=\argmax\nolimits_{z_i}\frac{\lambda(1-\beta)}{2}\sum\nolimits_{j\in \mathcal V}a_{ij}(1+\alpha)(1+x_i)\\
&-\frac{1}{2}(1-\lambda)\beta\sum\nolimits_{j \in \mathcal V}w_{ij}(y_i-1)^2-\frac{1}{2}\lambda\beta(y_i-x_i)^2.
\end{split}
\end{equation*}
Since $x_i$ can only take values $-1$ or $1$, the above equation is a discrete optimization problem. We can determine the optimal choice of $x_i$ by substituting these two values into $f_i(z_i, [\vect{1}_{n-1},\vect{1}_{n-1}])$. For $y_i\in[-1,1]$, we need to find the value that maximizes $f_i$ in the interval $[-1, 1]$.

\textbf{Case $x_i=+1$.} We compute
$f_i(z_i,[\vect{1}_{n-1},\vect{1}_{n-1}])=
    \lambda(1-\beta)(1+\alpha)\sum_{j\in \mathcal V}a_{ij}-\frac{1}{2}(1-\lambda)\beta\sum_{j \in \mathcal V}w_{ij}(y_i-1)^2-\frac{1}{2}\lambda\beta(y_i-1)^2$. %
Taking the partial derivative of $f_i$ with respect to $y_i$ and setting it to zero yields \begin{equation*}
\frac{\partial f_i}{\partial y_i} =- (1 - \lambda) \beta \sum\nolimits_{j \in \mathcal V} w_{ij} (y_i - 1) -\lambda \beta (y_i - 1)  = 0,
\end{equation*}
which is maximized by $y_i = 1$, which yields the claim.

\textbf{Case $x_i=-1$.} We compute  
$
f_i(z_i,[\vect{1}_{n-1},\vect{1}_{n-1}])=-\lambda(1-\beta)\sum_{j\in \mathcal V}a_{ij}-\frac{1}{2}(1-\lambda)\beta\sum_{j \in \mathcal V}w_{ij}(y_i+1)^2-\frac{1}{2}\lambda\beta(y_i+1)^2.
$
Taking the partial derivative of $f_i$ with respect to $y_i$ and setting it to zero, 
\begin{equation*}
    \frac{\partial f_i}{\partial y_i} = -(1 - \lambda) \beta \sum\nolimits_{j \in \mathcal V} w_{ij} (y_i +1) -\lambda \beta (y_i + 1) = 0,
\end{equation*}
which is maximized by $y_i=-1$.

In summary,  when $ \vect{z_{-i}} = [\vect{1}_{n-1}, \vect{1}_{n-1}] $, the best response is $[1,1]$ or $[-1,-1]$, depending on $x_i$, verifying that $[\vect{1}_{n}, \vect{1}_{n}] $ and $ [-\vect{1}_{n},-\vect{1}_{n}]$ are NE, using Lemma~\ref{lemma_Nash_equ}.
\qed

\section{Proof of Theorem~\ref{thm:basins of attraction_vp_vnp}}\label{app:thm:basins of attraction_vp_vnp}
First, our assumption on $\vect z(0)$ implies that $\mathcal V_1^+(0) \cup \mathcal V_{-1}^+(0) = \mathcal V$. Exploiting this, the key argument is to show that the following three statements hold true for all  $t\geq 0$: i) $\delta_i(\vect{z}(t))\geq0$, for all $i\in\mathcal V_1^+(0)$, and $\delta_i(\vect{z}(t))>0$, for all $i\in\mathcal V_{-1}^+(0)$; ii) $\vect x(t+1)\geq \vect x(t)$; and iii) $\vect{y}(t+1)\geq \vect 0_n$. Here, the inequalities are taken entry-wise. We prove the three statements jointly by (strong) induction. We start by proving that the statements hold true for $t=0$. To prove i), we compute $\delta_i$ by splitting each sum into two parts, obtaining
\begin{equation*}\label{eq:delta_t_equation_th1}
 \begin{aligned}
  \delta_i(\vect{z}(0)) &\hspace{-.07cm}=\hspace{-.07cm}  p \sum\nolimits_{j \in \mathcal{V}_{1}^+(0)}\hspace{-.1cm}w_{ij} y_j(0)\hspace{-.07cm}+\hspace{-.07cm}p \sum\nolimits_{k \in \mathcal{V}_{-1}^{+}(0)}\hspace{-.1cm} w_{i k} y_k(0) \\
              &+\frac{q}{4}\sum\nolimits_{j\in\mathcal{V}_{1}^{+}(0)}a_{ij}\left[2x_j(0)+\alpha(1+x_j(0))\right]\\
              &+\frac{q}{4}\sum\nolimits_{k\in \mathcal{V}^{+}_{-1}(0)}a_{ik}\left[2x_k(0)+\alpha(1+x_k(0))\right],\\
 \end{aligned} 
\end{equation*}
where $p=2\lambda\beta(1-\lambda)$ and $q=2\lambda(1-\beta)$ are positive constants. For $i\in \mathcal{V}_{1}^+(0)$, we use the cohesiveness of $\mathcal{V}_{1}^+(0)$ and the fact that $\vect{y}(0)\geq 0_n$, to bound
\begin{equation}\label{eq:delta_t_equation2_th1}
 \begin{aligned}
  \delta_i(\vect{z}(0)) & \geq\frac{q}{2}\sum\nolimits_{j\in\mathcal{V}_{1}^{+}(0)}\hspace{-.1cm}a_{ij}(1+\alpha)\hspace{-.07cm}-\hspace{-.07cm}\frac{q}{2}\sum\nolimits_{k\in \mathcal{V}^{+}_{-1}(0)}\hspace{-.1cm}a_{ik}\\
  & \geq\frac{q}{2}\left[\frac{1}{\alpha+2}(1+\alpha)-\frac{\alpha+1}{\alpha+2}\right]\geq 0.
 \end{aligned} 
\end{equation}
Similarly, for $i\in \mathcal{V}_{-1}^{+}(0)$, we use the diffusiveness of $\mathcal {V}_{-1}^{+}(0)$ and the fact that $\vect{y}(0)\geq 0$, to bound
\begin{equation}\label{eq:delta_t_equation3_th1}
 \begin{aligned}
  \delta_i(\vect{z}(0)) & \geq\frac{q}{2}\sum\nolimits_{j\in\mathcal{V}_{1}^{+}(0)}\hspace{-.07cm}a_{ij}(1+\alpha)\hspace{-.07cm}-\hspace{-.07cm}\frac{q}{2}\sum\nolimits_{k\in \mathcal{V}^{+}_{-1}(0)}\hspace{-.07cm}a_{ik}\\
  & >\frac{q}{2}\left[\frac{1}{\alpha+2}(1+\alpha)-\frac{\alpha+1}{\alpha+2}\right]\geq 0.
 \end{aligned} 
\end{equation}
The strict inequality is obtained due to the $\frac{\alpha+1}{\alpha+2}$-diffusiveness of $\mathcal {V}_{-1}^+(0)$, which implies that $\sum\nolimits_{k\in \mathcal{V}^{+}_{-1}(0)}a_{ik}<\frac{\alpha+1}{\alpha+2}$. To prove ii) and iii), let $i_0$ be the agent active at time $t=0$. For all $i\neq i_0$, no update is performed, hence, it clearly holds true that $x_i(1)=x_i(0)$ and $y_i(1)=y_i(0)\geq 0$. Then, we distinguish two cases. If $i_0\in\mathcal V_1^+(0)$, and because $\delta_i(\vect z(0))\geq 0$ as established above, it follows that $x_i(1)=1$ and $y_i(1)\geq\lambda\geq 0$, according to Proposition~\ref{prop:dyn}. If $i_0\in\mathcal V_{-1}^+(0)$, and because $\delta_i(\vect z(0))> 0$ as established above, it follows that $x_i(1)=1$ and $y_i(1)\geq\lambda\geq 0$, according to Proposition~\ref{prop:dyn}, yielding the claim. This concludes the proof of the induction basis.
We are now left to prove that, assuming i)--iii) hold up to time $t$, they necessarily hold at $t+1$. The proof of i) follows the same arguments used to prove the induction basis. In fact, due to  strong induction assumption, we have $x_i(t)\geq x_i(0)$, for all $i\in\mathcal V$, implying
\begin{equation}\label{eq:delta_t_equation4}
 \begin{aligned}
  \delta_i(\vect{z}(t))  &\geq \frac{q}{4}\sum\nolimits_{j\in\mathcal{V}_{1}^{+}(t)}a_{ij}\left[2x_j(t)+\alpha(1+x_j(t))\right]\\
              &+\frac{q}{4}\sum\nolimits_{k\in \mathcal{V}^{+}_{-1}(t)}a_{ik}\left[2x_k(t)+\alpha(1+x_k(t))\right],\\
           \geq   & \frac{q}{4}\sum\nolimits_{j\in\mathcal{V}_{1}^{+}(0)}a_{ij}\left[2x_j(0)+\alpha(1+x_j(0))\right]\\
              &+\frac{q}{4}\sum\nolimits_{k\in \mathcal{V}^{+}_{-1}(0)}a_{ik}\left[2x_k(0)+\alpha(1+x_k(0))\right],
 \end{aligned} 
\end{equation}
which is nonnegative (or strictly positive) as per \eqref{eq:delta_t_equation2_th1}
and \eqref{eq:delta_t_equation3_th1}, respectively. Similar, ii) and iii) are obtained as consequences of the positivity of $\delta_i$ from \eqref{eq:delta_t_equation4}.

Finally, if at a generic time $t$, agent $i_t\in\mathcal V_{-1}^{+}(0)$ activates, then $\delta_{i_t}(\vect{z}(t))>0$, implying that $x_{i_t}(t+1)=1$. Due to Assumption~\ref{as:activation} and the monotonicity property proved above in item ii), it necessarily holds true that $\vect{x}(t)=\vect{1}$, for all $t\geq T$. Finally, the fact that $[\vect{1}_n,\vect{1}_n]$ is the only NE with $\vect{x}=\vect{1}_n$ guarantees convergence~\cite[Theorem 1]{raineri2025}. 

\qed
\section{Proof of Theorem~\ref{thm:basins of attraction_vnvp_consensus}}\label{app:thm:basins of attraction_vnvp_consensus}
For the sake of convenience in our proof, we define positive constants $p=2\lambda\beta(1-\lambda)$ and $q=2\lambda(1-\beta)$.
The proof is based on the following set of arguments. If the partitions $\mathcal{V}_{1}^{+}(0)$ and $\mathcal{V}_{-1}^{-}(0)$ satisfy
\vspace{-0.1cm}
\begin{subequations}
\begin{align}
\sum\nolimits_{j\in\mathcal{V}_{1}^{+}(0)}w_{ij} &\geq 1-\frac{q}{2(p+q)}, && i\in\mathcal{V}_{1}^{+}(0), \label{eq:xcoo_cond_pp} \\
\sum\nolimits_{j\in\mathcal{V}_{-1}^{-}(0)}w_{ij}  &< \frac{q}{2(p+q)}, && i\in\mathcal{V}_{-1}^{-}(0),\label{eq:xcoo_cond_nn}
\end{align}
\end{subequations}
then i) $\delta_i(\vect{z}(t))\geq0$, for all $i\in\mathcal V_1^+(0)$, and $\delta_i(\vect{z}(t))>0$, for all $i\in\mathcal V_{-1}^{-}(0)$; and ii) the action vector $\vect x(t+1)\geq \vect x(t)$.  Then, we show that if there also holds 
\vspace{-0.1cm}
\begin{subequations}
\begin{align}
 \sum\nolimits_{j\in\mathcal{V}_{1}^{+}(0)}w_{ij}\geq\frac{1-2\lambda}{1-\lambda} && i\in\mathcal{V}_{1}^{+}(0),\label{eq:ycoo_cond_pp} \\
\sum\nolimits_{j\in\mathcal{V}_{-1}^{-}(0)}w_{ij}< \frac{\lambda}{1-\lambda},&& i\in\mathcal{V}_{-1}^{-}(0),\label{eq:ycoo_cond_nn} 
\end{align}
\end{subequations}
we have iii) the opinion vector $\vect{y}(t+1))\geq \vect 0_n$. 
We start by proving that the statements hold for $t=0$. To prove i), we compute $\delta_i$ by splitting each sum into two parts, obtaining
\vspace{-0.1cm}
\begin{equation*}\label{eq:deltacon_0001_vp_pn}
 \begin{aligned}
  \delta_i&(\vect{z}(0)) \hspace{-.07cm}=\hspace{-.07cm}  p \sum\nolimits_{j \in \mathcal{V}_{1}^+(0)}\hspace{-.1cm}w_{ij} y_j(0)\hspace{-.07cm}+\hspace{-.07cm}p \sum\nolimits_{k \in \mathcal{V}_{-1}^{-}(0)}\hspace{-.1cm} w_{i k} y_k(0) \\
              &+\frac{q}{4}\left[\sum\nolimits_{j\in\mathcal{V}_{1}^{+}(0)}w_{ij}2x_j(0)+\sum\nolimits_{k\in \mathcal{V}^{-}_{-1}(0)}w_{ik}2x_k(0)\right].
 \end{aligned} 
 \end{equation*}
For $i\in \mathcal{V}_{1}^+(0)$, we use \eqref{eq:xcoo_cond_pp},  $\sum_{k\in \mathcal{V}_{-1}^{-}(0)}w_{ik} = 1-\sum_{j\in \mathcal{V}_{1}^{+}(0)}w_{ij}$ because $W$ is stochastic, and $y_k(0)\in[-1,0),\forall k\in\mathcal{V}_{-1}^{-}(0),y_j(0)\in(0,+1],\forall j\in\mathcal{V}_{1}^{+}(0)$ to bound
\vspace{-0.1cm}
\begin{equation*}\label{eq:deltacon_03_vp_pn}
 \begin{aligned}
 \delta_i(\vect{z}(0)) 
 &\geq -p(1\hspace{-.07cm}-\hspace{-.07cm}\sum\nolimits_{j\in \mathcal{V}_{1}^{+}(0)}w_{ij})+q\Big[\hspace{-0.05cm}\sum\nolimits_{j\in \mathcal{V}_{1}^{+}(0)}w_{ij} \hspace{-.07cm}-\hspace{-.07cm}\frac12\Big]\\
 &\geq  -(p+\frac{q}{2})+(p+q)\sum\nolimits_{j\in \mathcal{V}_{1}^{+}(0)}w_{ij}\\
 &\geq  -(p+\frac{q}{2})+(p+q)\big(1-\frac{q}{2(p+q)}\big)\geq 0.
  \end{aligned} 
\end{equation*}
Similarly, for $i\in \mathcal{V}_{-1}^{-}(0)$, we use \eqref{eq:xcoo_cond_nn}, $\sum_{j\in \mathcal{V}_{1}^{+}(0)}w_{ij}=1-\sum_{k\in \mathcal{V}_{-1}^{-}(0)}w_{ik}$, and $y_k(0)\in[-1,0),\forall k\in\mathcal{V}_{-1}^{-}(0),y_j(0)\in(0,+1],\forall j\in\mathcal{V}_{1}^{+}(0)$ to bound
\vspace{-0.1cm}
\begin{equation*}\label{eq:delta_t_equation3_th2}
 \begin{aligned}
 \delta_i(\vect{z}(0)) &\geq -p\hspace{-.1cm}\sum\nolimits_{k\in \mathcal{V}_{-1}^{-}(0)}w_{ik}+q\Big[\frac12 -\sum\nolimits_{k\in \mathcal{V}_{-1}^{-}(0)}w_{ik}\Big]\\
 &> -(p+q)\sum\nolimits_{k\in \mathcal{V}_{-1}^{-}(0)}w_{ik}+\frac{q}{2}\\
  &> -(p+q)\frac{q}{2(p+q)}+\frac{q}{2}\geq 0.
 \end{aligned} 
\end{equation*}
To prove ii) and iii), let $i_0$ be the agent active at time $t=0$. For all $i\neq i_0$, it holds  $x_i(1)=x_i(0)$ and $y_i(1)=y_i(0)\geq 0$. Then, we distinguish two cases. If $i_0\in\mathcal V_{1}^{+}(0)$, being $\delta_i(\vect z(0))\geq 0$, then $x_i(1)=1$ and $y_i(1)\geq(1-\lambda)(-\sum_{k\in \mathcal{V}_{-1}^{-}(0)}w_{ik} )+\lambda\geq 0$
where we used \eqref{eq:ycoo_cond_pp} and $\sum_{j\in \mathcal{V}_{1}^{+}(0)}w_{ij}+ \sum_{k\in \mathcal{V}_{-1}^{-}(0)}w_{ik} =1 $. If $i_0\in\mathcal V_{-1}^{-}(0)$, being $\delta_i(\vect z(0))> 0$, then $x_i(1)=1$ and $y_i(1)\geq(1-\lambda)(-\sum\nolimits_{k\in \mathcal{V}_{-1}^{-}(0)}w_{ik} )+\lambda>0$ according to \eqref{eq:ycoo_cond_nn}, concluding the proof of the induction basis. 

Finally, by using a strong induction argument, similar to the proof of Theorem~\ref{thm:basins of attraction_vp_vnp}, it can also be established that, assuming that the three properties hold up to time $t$, they necessarily hold true at time $t+1$, yielding convergence. 
\qed
\section{Proof of Theorem~\ref{thm:anti_convergence_all}}\label{app:thm:anti_convergence_all}
The proof primarily uses game-theoretic methods. We first introduce the following function:
\begin{equation}\label{mixed_Phi(z)}
 \begin{aligned}
  &\Phi(z)=-\sum\nolimits_{i\in\mathcal{V}}\sum\nolimits_{j\in\mathcal{V}\backslash \{i\}}\eta_i\frac{a_{ij}}{2}\Big[(1-x_j)(1-x_i)\\
  &-(1+\alpha)(1+x_j)(1+x_i)\Big]-\frac{1}{2}\sum_{i\in\mathcal{V}}\sum_{j\in\mathcal{V}}\frac{ w_{ij}}{2}(y_i-y_j)^2\\[-6pt]
  &-\frac{1}{2}\sum\nolimits_{k\in\mathcal{V}}\frac{\lambda_i}{(1-\lambda_i)
  }(x_k-y_k)^2,
 \end{aligned} 
\end{equation} 
with $\eta_i=\lambda_i(1-\beta_i)/4\beta_i(1-\lambda_i)$. 
Following an approach similar to~\cite[Lemma 1]{aghbolagh2023coevolutionary} (and hence details are omitted), one can prove that, under Assumption~\ref{as:uniform} and when $\varepsilon=-1$, the coevolutionary game is a generalized ordinal potential game~\cite{monderer1996potential} with the potential function in \eqref{mixed_Phi(z)}. 
Following the approach of \cite[Theorem~1]{aghbolagh2023coevolutionary}, we can prove that the generalized ordinal potential game necessarily converges to an equilibrium, yielding the claim. The details are omitted.
\qed
\section{Proof of Theorem~\ref{thm:consensus_ equilibrium_anti}}\label{app:consensus_ equilibrium_anti}
 Our objective is to identify conditions for when the state $\vect z^*=(\vect x^*,\vect y^*)=[\vect{1}_n,\vect{1}_n]$ (or $-[\vect{1}_n,\vect{1}_n]$) is an equilibrium of the coevolutionary model on the two-layer network $\gG$. If  $[\vect{1}_n,\vect{1}_n]$ is an equilibrium, then from~\eqref{update_rule_coo_xy}, there must hold
\begin{subequations}
\begin{align}
    1=x_i^*&=\mathcal{S}(\delta_i(z^*)),\label{anti_action_update_opti}\\ 
    1=y_i^*&=(1-\lambda)\sum\nolimits_{j\in\mathcal{V}}w_{ij}y_j^{*}+\lambda\mathcal{S}(\delta_i(z^*)).\label{anti_opinion_update_opti}  
\end{align}
\end{subequations}
\eqref{anti_opinion_update_opti} can be further simplified to $1= (1-\lambda)\sum\nolimits_{j\in\mathcal{V}}w_{ij}+\lambda$, which always holds since $W$ is row-stochastic. Meanwhile, \eqref{delta_i(z)} evaluated at $[\vect{1}_n,\vect{1}_n]$ yields \[\delta_i(\vect z^*) = -\frac{\lambda(1-\beta)}{2} \sum_{j\in \mathcal V}a_{ij}(2+2\alpha) +2(1-\lambda)\lambda\beta\sum\nolimits_{j \in \mathcal V}w_{ij}.\]
It follows that $\delta_i(\vect z^*) > 0$ (which ensures $x_i^* =\mathcal{S}(\delta_i(\vect z^*)) = 1$) is implied by $\sum_{j\in \mathcal V}a_{ij} <\frac{2(1-\lambda)\beta}{(1-\beta)(1+\alpha)}$. Thus, $[\vect{1}_n,\vect{1}_n]$ satisfies the equilibrium equations if \eqref{eq:condition_A} holds for all $i\in\mathcal V$.
\qed

\section{Proof of Theorem~\ref{thm:polarized_anti}}\label{app:thm:polarized_anti}
The proof is divided into two parts. In the first part, we establish that given a fixed action vector $\vect{x}^*$ that is polarized concerning a partitioning $(\mathcal{V}_p,\mathcal{V}_n)$, and under the conditions in \eqref{eq:pol_y_cond_anti}, the unique vector $\vect{y}^*$ that satisfies the equilibrium conditions for \eqref{opinion_update} is polarized. The proof that $\vect{y}^*$ is polarizing with respect to $(\mathcal{V}_p,\mathcal{V}_n)$ follows closely the proof of~\cite[Theorem~2]{aghbolagh2023coevolutionary}, and it is thus omitted. 

In the second part, with $\vect y^*$ being the polarized vector identified in part one, we show that conditions in \eqref{eq:pol_xp_cond_anti} and \eqref{eq:pol_xn_cond_anti}, guarantee that the action vector $\vect{x}^*$  is invariant under \eqref{action_update}, which implies that $\vect{z}^*$ is a polarized equilibrium of the coevolutionary model. Given $\vect{y}^*$ which is polarized with respect to $(\mathcal{V}_p,\mathcal{V}_n)$, define $y_m^+=\min_{i\in\mathcal{V}_p } y_i^*$ and $y_m^-=\min_{i\in\mathcal{V}_n}$ and $y_M^+ = \max_{i\in\mathcal{V}_p} y_i^*$ and $y_M^- = \max_{i\in\mathcal{V}_n} y_i^*$.

For any $i\in\mathcal{V}_p $ and $t\geq 0$, $x_i(t) = +1 = x_i(t+1) = +1$ if $\delta_i(\vect{z}(t)) \geq 0$. From \eqref{eq:delta_ass1_4_coo}, we obtain the following bound
\begin{align}
    \delta_i(\vect{z})&=p\Big[\sum\nolimits_{j\in \mathcal{V}_p }w_{ij}y_j+\sum\nolimits_{k\in \mathcal{V}_n}w_{ik}y_k\Big] \nonumber \\
   &-\frac{q}{2}\Big[\sum\nolimits_{j\in\mathcal{V}_p }w_{ij}x_j+\sum\nolimits_{k\in\mathcal{V}_n}w_{ik}x_k\Big] \nonumber \\ 
   &\geq p\Big[\underline{d}_p y_m^++(1-\underline{d}_p )y_m^-\Big]-\frac{q}{2}\Big[\bar{d}_p -(1-\bar{d}_p )\Big],\label{eq:delta_cond_a}
\end{align}
where $p=2\lambda\beta(1-\lambda), q=2\lambda (1-\beta)$. We shall now prove that 
\be\label{eq:deltaTheoremPol}
2\lambda\beta(1-\lambda)(\bar{d}_p  y_m^+ +(1-\bar{d}_p )y_m^-)+(1-\beta)\lambda(1-2\bar{d}_p)>0,
\ee
which implies that \eqref{eq:delta_cond_a} holds. Before proceeding, we state two important inequalities, whose derivations are found in the proof of ~\cite[Theorem~2]{aghbolagh2023coevolutionary}, and is obtained by leveraging the equilibrium expression for $y^*_i$, i.e., setting $y_i(t+1) = y_i(t)$ in \eqref{update_rule_coo_xy}. In particular, one can show that $y_m^+\geq \lambda+(1-\lambda)[\underline{d}_p y_m^+ +(1-\underline{d}_p )y_m^-]$ and also that
\begin{equation}\label{eq:yPlusInEqFinal_1_new}
y_m^+\geq\frac{\lambda-\lambda(1-\lambda)(1-\underline{d}_p+\bar{d}_n )}{1-(1-\lambda)(\underline{d}_p+\bar{d}_n )-(1-\lambda)^2(1-\underline{d}_p-\bar{d}_n )}.
\end{equation}
The first inequality can be used to show that $2\beta(y_m^+-\lambda)-(1-\beta)(2\bar{d}_p -1)>0$ implies \eqref{eq:deltaTheoremPol}. Substituting $y^+_m$ from \eqref{eq:yPlusInEqFinal_1_new} into this latest inequality, and simplifying, yields \eqref{eq:pol_xp_cond_anti}. Thus, \eqref{eq:pol_xp_cond_anti} implies \eqref{eq:delta_cond_a}, and thus $\delta_i(\vect z) > 0$. In other words, $x_i(t+1) = x_i(t) = +1$ for all $t\geq 0$ if \eqref{eq:pol_xp_cond_anti} holds.

For any $i\in\mathcal{V}_n$ and $t\geq 0$, $ x_i(t+1) = x_i(t) = -1$ iff $\delta_i(x(t),y(t)) < 0$, which is implied by
\begin{equation}
\begin{array}{l}
\delta_i\displaystyle(\vect{z})=p\sum\nolimits_{j\in \mathcal{V}_p }w_{ij}y_j+p\sum\nolimits_{k\in \mathcal{V}_n}w_{ik}y_k  \\\quad
   \displaystyle-\frac{q}{2}\Big[\sum\nolimits_{j\in\mathcal{V}_p }w_{ij}x_j+\sum\nolimits_{k\in\mathcal{V}_n}w_{ik}x_k\Big]\\
   \quad\displaystyle\leq p\Big[\bar{d}_n  y_M^-+(1-\bar{d}_n )y_M^+\Big]-\frac{q}{2}\Big[(1-\bar{d}_n )-\bar{d}_n \Big], \label{eq:delta_vp_cond_a}
\end{array}
\end{equation}
where $p=2\lambda\beta(1-\lambda), q=2\lambda (1-\beta)$. We omit the detailed arguments, which follow similarly to the above, one can show that \eqref{eq:pol_xn_cond_anti} implies the right hand side of \eqref{eq:delta_vp_cond_a} is strictly negative, and thus $\delta_i(\vect z) < 0$.
\qed

\section{Proof of Theorem~\ref{thm:polarized_convergence_anti}}\label{app:polarized_convergence_anti}
First, our assumption on polarized $\vect z(0)$ implies that $\mathcal V_1^+(0) \cup \mathcal V_{-1}^{-}(0) = \mathcal V$. 
For the sake of convenience in our proof, we define positive constants $p=2\lambda\beta(1-\lambda)$ and $q=2\lambda(1-\beta)$.
The proof is by induction. To begin, we show that if the partitions $\mathcal{V}_{1}^{+}(0)$ and $\mathcal{V}_{-1}^{-}(0)$ satisfy \eqref{eq:condition_polarized_anti}, then i) $\vect{x}(1)$ is polarized with respect to the same partition and ii) $\vect{y}(1)$ is polarized with respect to the same partition, i.e., $\mathcal{V}_{1}^{+}(1) = \mathcal{V}_{1}^{+}(0)$ and $\mathcal{V}_{-1}^{-}(1) = \mathcal{V}_{-1}^{-}(0)$
Then, Theorems~\ref{thm:anti_convergence_all} and~\ref{thm:polarized_anti} are leveraged to guarantee convergence to equilibrium , and that $\mathcal{V}_{1}^{+}(t) = \mathcal{V}_{1}^{+}(0)$ and $\mathcal{V}_{-1}^{-}(t) = \mathcal{V}_{-1}^{-}(0)$ for all $t\geq 0$ (implying the equilibrium is polarized), respectively.

We start from the induction basis. To prove i), we compute $\delta_i(\vect{z}(0))$ by splitting each sum into two parts, obtaining 
\begin{equation}\begin{array}{l}\displaystyle\delta_i(\vect{z}(0)) \hspace{-.07cm}=\hspace{-.07cm}  p\Big[\sum\nolimits_{j \in \mathcal{V}_{1}^+(0)}\hspace{-.1cm}w_{ij} y_j(0)\hspace{-.07cm}+\hspace{-.07cm}\sum\nolimits_{k \in \mathcal{V}_{-1}^{-}(0)}\hspace{-.1cm} w_{i k} y_k(0)\Big]\\ \quad\displaystyle-\frac{q}{2}\Big[\sum\nolimits_{j\in\mathcal{V}_{1}^{+}(0)}w_{ij}x_j(0)+\sum\nolimits_{k\in \mathcal{V}^{-}_{-1}(0)}w_{ik}x_k(0)\Big].
\end{array}\end{equation}
For $i\in \mathcal{V}_{1}^+(0)$, we use the following facts: a) the inequalities in \eqref{eq:condition_polarized_anti}, b) $\beta<1/(2-\lambda)$, c) $y_k(0)\in[-1,0),\forall k\in\mathcal{V}_{-1}^{-}(0),y_j(0)\in(0,+1],\forall j\in\mathcal{V}_{1}^{+}(0)$, and d) $\sum_{k\in \mathcal{V}_{-1}^{-}(0)}w_{ik} = 1-\sum_{j\in \mathcal{V}_{1}^{+}(0)}w_{ij}$ to bound
\vspace{-0.1cm}
\begin{equation}\label{eq:deltacon_03_vp_pn_polar}
 \begin{aligned}
\delta_i&(\vect{z}(0))=p\left[\sum\nolimits_{j\in \mathcal{V}_{+1}^{+}(0)}w_{ij}y_j+\right.
\left.\hspace{-.07cm}\sum\nolimits_{k\in \mathcal{V}^{-}_{-1}(0)}w_{ik}y_k\right]\\
  & \quad \ -q\Big( \sum\nolimits_{j\in\mathcal{V}_{+1}^{+}(0)}w_{ij}-\frac12\Big)\\
  &\geq -p\sum\nolimits_{k\in \mathcal{V}^{-}_{-1}(0)}w_{ik}-q\sum\nolimits_{j\in\mathcal{V}_{+1}^{+}(0)}w_{ij} +\frac{q}{2}\\
   &\geq (p-q)\sum\nolimits_{j\in \mathcal{V}_{+1}^{+}(0)}w_{ij}-p+\frac{q}{2}\\
  &\geq (p-q) \frac12\big(1+\frac{p}{p-q}\big)+\big(\frac{q}{2} +p\big)> 0.
\end{aligned}
\end{equation}
For $i\in \mathcal{V}_{-1}^{-}(0)$, we use the inequalities in \eqref{eq:condition_polarized_anti}, the condition $\beta<1/(2-\lambda)$, the fact that $y_k(0)\in[-1,0),\forall k\in\mathcal{V}_{-1}^{-}(0),y_j(0)\in(0,+1],\forall j\in\mathcal{V}_{1}^{+}(0)$, and the equality $\sum_{j\in \mathcal{V}_{1}^{+}(0)}w_{ij}=1-\sum_{k\in \mathcal{V}_{-1}^{-}(0)}w_{ik}$ to bound
\vspace{-0.1cm}
\begin{equation}\label{eq:deltacon_04_vp_pn_polar}
 \begin{aligned}
\delta_i&(\vect{z}(0))=p\left[\sum\nolimits_{j\in \mathcal{V}_{+1}^{+}(0)}w_{ij}y_j\right.
\left.\hspace{-.07cm}+\sum\nolimits_{k\in \mathcal{V}^{-}_{-1}(0)}w_{ik}y_k\right]\\
  & \quad \ -q\Big(\frac12 -\hspace{-.1cm}\sum\nolimits_{j\in\mathcal{V}_{-1}^{-}(0)}w_{ij}\Big)\\
  &\leq p\sum\nolimits_{j\in \mathcal{V}_{+1}^{+}(0)}w_{ij}-q\Big(\frac12 -\sum\nolimits_{k\in \mathcal{V}^{-}_{-1}(0)}w_{ik}\Big)\\
  &\leq -(p-q)\sum\nolimits_{k\in \mathcal{V}^{-}_{-1}(0)}w_{ik}+\big(p-\frac{q}{2}\big)\\
  &\leq -(p-q) \frac12\big(1+\frac{p}{p-q}\big)+\big(p-\frac{q}{2}\big)< 0.
\end{aligned}
\end{equation}
Let $i_0$ be the agent active at time $t=0$. For all $i\neq i_0$, no update is performed, so those agents in $\mathcal{V}^-_{-1}(0)$ remain in $\mathcal{V}^-_{-1}(1)$ (and agents in $\mathcal{V}^+_{+1}(0)$ stay in $\mathcal{V}^+_{+1}(1)$ ). Concerning $i_0$, we distinguish two cases. If $i_0\in\mathcal V_{1}^{+}(0)$, being $\delta_i(\vect z(0))\geq 0$, then $x_i(1)=x_i(0)=1$. From \eqref{opinion_update}, we obtain $y_i(1)\geq(1-\lambda)(-\sum\nolimits_{k\in \mathcal{V}_{-1}^{-}(0)}w_{ik} )+\lambda>0,$ with the latter due to the left inequality of \eqref{eq:condition_polarized_anti}. If $i_0\in\mathcal V_{-1}^{-}(0)$, being $\delta_i(\vect z(0))< 0$, then $x_i(1)=-1$ and similarly, $y_i(1)\leq(1-\lambda)\sum\nolimits_{j\in \mathcal{V}_{+1}^{+}(0)}w_{ij}-\lambda<0$ according to \eqref{eq:condition_polarized_anti}. This concludes the proof of the induction basis. 
By using an induction argument, it can be shown by repeating the arguments used in \eqref{eq:deltacon_03_vp_pn_polar} and \eqref{eq:deltacon_04_vp_pn_polar} that if $\vect{z}(t)$ is polarized with respect to the same partitioning and \eqref{eq:condition_polarized_anti} holds, then also $\vect{z}(t+1)$ is polarized with respect to that partitioning. Hence, the induction principle yields the first claim. 


Finally, Theorem~\ref{thm:anti_convergence_all} guarantees  convergence to an equilibrium: $\lim_{t\to\infty} \vect{z}(t)\to\vect{z}^*$. Now, $\vect z(0)$ is polarized with respect to $(\mathcal{V}_{1}^{+}(0),\mathcal{V}_{-1}^{-}(0))$ by hypothesis, and the above has established that $\mathcal{V}_{1}^{+}(t) = \mathcal{V}_{1}^{+}(0)$ and $\mathcal{V}_{-1}^{-}(t) = \mathcal{V}_{-1}^{-}(0)$ for all $t\geq 0$. In other words, the elements in the two sets remain unchanged over time. Separately, the inequalities in \eqref{eq:condition_polarized_anti} satisfy the conditions~(\ref{eq:pol_y_cond_anti})--(\ref{eq:pol_xn_cond_anti}) of Theorem~\ref{thm:polarized_anti} (which can be applied since Assumptions~\ref{as:uniform}--\ref{as:polarization} hold), and thus $\vect z^*$ must be polarized with respect to $(\mathcal{V}_{1}^{+}(0),\mathcal{V}_{-1}^{-}(0))$.
For notational convenience, we denote $\mathcal{V}_{1}^{+}(0)$ and $\mathcal{V}_{-1}^{-}(0)$ as $\mathcal{V}_p$ and $\mathcal{V}_n$, which completes the proof.
\qed

\bla

\bibliographystyle{IEEEtran}
\bibliography{Coevo}

\begin{thebibliography}{10}
\providecommand{\url}[1]{#1}
\csname url@samestyle\endcsname
\providecommand{\newblock}{\relax}
\providecommand{\bibinfo}[2]{#2}
\providecommand{\BIBentrySTDinterwordspacing}{\spaceskip=0pt\relax}
\providecommand{\BIBentryALTinterwordstretchfactor}{4}
\providecommand{\BIBentryALTinterwordspacing}{\spaceskip=\fontdimen2\font plus
\BIBentryALTinterwordstretchfactor\fontdimen3\font minus
  \fontdimen4\font\relax}
\providecommand{\BIBforeignlanguage}[2]{{%
\expandafter\ifx\csname l@#1\endcsname\relax
\typeout{** WARNING: IEEEtran.bst: No hyphenation pattern has been}%
\typeout{** loaded for the language `#1'. Using the pattern for}%
\typeout{** the default language instead.}%
\else
\language=\csname l@#1\endcsname
\fi
#2}}
\providecommand{\BIBdecl}{\relax}
\BIBdecl

\bibitem{french1956_socialpower}
J.~R.~P. French~Jr, ``A formal theory of social power,'' \emph{Psychol. Rev.},
  vol.~63, no.~3, pp. 181--194, 1956.

\bibitem{hegselmann2002opinion}
R.~Hegselmann and U.~Krause, ``{Opinion dynamics and bounded confidence models,
  analysis, and simulation},'' \emph{J. Artific. Soc. Soc. Simul.}, vol.~5,
  no.~3, pp. 1--30, 2002.

\bibitem{biswas2009model}
S.~Biswas and P.~Sen, ``Model of binary opinion dynamics: Coarsening and effect
  of disorder,'' \emph{Phys. Rev. E}, vol.~80, no.~2, pp. 4--7, 2009.

\bibitem{ding2010evolutionary}
F.~Ding, Y.~Liu, B.~Shen, and X.-M. Si, ``An evolutionary game theory model of
  binary opinion formation,'' \emph{Physica A}, vol. 389, no.~8, pp.
  1745--1752, 2010.

\bibitem{friedkin2011social_book}
N.~E. Friedkin and E.~C. Johnsen, \emph{{Social Influence Network Theory: A
  Sociological Examination of Small Group Dynamics}}.\hskip 1em plus 0.5em
  minus 0.4em\relax Cambridge University Press, 2011.

\bibitem{proskurnikov2017tutorial}
A.~V. Proskurnikov and R.~Tempo, ``{A tutorial on modeling and analysis of
  dynamic social networks. Part I},'' \emph{Annu. Rev. Control}, vol.~43, pp.
  65--79, 2017.

\bibitem{ravazzi2021learning}
C.~Ravazzi, F.~Dabbene, C.~Lagoa, and A.~V. Proskurnikov, ``Learning hidden
  influences in large-scale dynamical social networks: A data-driven
  sparsity-based approach, in memory of roberto tempo,'' \emph{IEEE Control
  Syst.}, vol.~41, no.~5, pp. 61--103, 2021.

\bibitem{Loy2022}
N.~Loy, M.~Raviola, and A.~Tosin, ``Opinion polarization in social networks,''
  \emph{Philos. Trans. R. Soc. A}, vol. 380, no. 2224, 2022.

\bibitem{friedkin1990social}
N.~E. Friedkin and E.~C. Johnsen, ``Social influence and opinions,'' \emph{J.
  Math. Sociol.}, vol.~15, no. 3-4, pp. 193--206, 1990.

\bibitem{altafini2012consensus}
C.~Altafini, ``{Consensus Problems on Networks with Antagonistic
  Interactions},'' \emph{IEEE Trans. Automat. Contr.}, vol.~58, no.~4, pp.
  935--946, 2013.

\bibitem{bolzern2019opinion}
P.~Bolzern, P.~Colaneri, and G.~De~Nicolao, ``Opinion influence and evolution
  in social networks: A markovian agents model,'' \emph{Automatica}, vol. 100,
  pp. 219--230, 2019.

\bibitem{acemoglu2011opinion}
D.~Acemoglu, G.~Como, F.~Fagnani, and A.~Ozdaglar, ``Opinion fluctuations and
  persistent disagreement in social networks,'' in \emph{50th IEEE Conf. Decis.
  Control}, 2011, pp. 2347--2352.

\bibitem{Velarde2020Polarization}
P.~Cisneros-Velarde, K.~Chan, and F.~Bullo, ``Polarization and fluctuations in
  signed social networks,'' \emph{IEEE Trans. Automat. Contr.}, vol.~66, no.~8,
  p. 3789–3793, 2020.

\bibitem{shi2022}
L.~Shi, Y.~Cheng, J.~Shao, X.~Wang, and H.~Sheng, ``Leader-follower opinion
  dynamics of signed social networks with asynchronous trust/distrust level
  evolution,'' \emph{IEEE Trans. Netw. Sci. Eng.}, vol.~9, no.~2, p. 495–509,
  2022.

\bibitem{DePasquale2022}
G.~De~Pasquale and M.~E. Valcher, ``Consensus for clusters of agents with
  cooperative and antagonistic relationships,'' \emph{Automatica}, vol. 135, p.
  110002, 2022.

\bibitem{Bizyaeva2023}
A.~Bizyaeva, A.~Franci, and N.~E. Leonard, ``Nonlinear opinion dynamics with
  tunable sensitivity,'' \emph{IEEE Trans. Autom. Control.}, vol.~68, no.~3, p.
  1415–1430, 2023.

\bibitem{godes2009firm}
D.~Godes and D.~Mayzlin, ``Firm-created word-of-mouth communication: Evidence
  from a field test,'' \emph{Mark. Sci.}, vol.~28, no.~4, pp. 721--739, 2009.

\bibitem{candia2009advertising}
J.~Candia, ``Advertising and irreversible opinion spreading in complex social
  networks,'' \emph{Int. J. Mod. Phys. C}, vol.~20, no.~06, pp. 799--815, 2009.

\bibitem{castro2017opinion}
J.~Castro, J.~Lu, G.~Zhang, Y.~Dong, and L.~Mart{\'\i}nez, ``Opinion
  dynamics-based group recommender systems,'' \emph{IEEE Trans. Syst. Man
  Cybern.: Syst.}, vol.~48, no.~12, pp. 2394--2406, 2017.

\bibitem{tang2013social}
J.~Tang, X.~Hu, and H.~Liu, ``Social recommendation: a review,'' \emph{Soc.
  Netw. Anal. Min.}, vol.~3, pp. 1113--1133, 2013.

\bibitem{She2022}
B.~She, J.~Liu, S.~Sundaram, and P.~E. Paré, ``On a networked sis epidemic
  model with cooperative and antagonistic opinion dynamics,'' \emph{IEEE Trans.
  Control Netw. Syst.}, vol.~9, no.~3, p. 1154–1165, 2022.

\bibitem{Xu2024}
Q.~Xu and H.~Ishii, ``On a discrete-time networked siv epidemic model with
  polar opinion dynamics,'' \emph{IEEE Trans. Netw. Sci. Eng.}, vol.~11, no.~6,
  p. 6636–6651, 2024.

\bibitem{Jackson2015}
M.~O. Jackson and Y.~Zenou, ``{Games on Networks},'' in \emph{{Handbook of Game
  Theory with Economic Applications}}.\hskip 1em plus 0.5em minus 0.4em\relax
  Elsevier, 2015, vol.~4, ch.~3, pp. 95--163.

\bibitem{young2011dynamics}
H.~P. Young, ``The dynamics of social innovation,'' \emph{Proc. Natl. Acad.
  Sci. USA}, vol. 108, no. Supplement 4, pp. 21\,285--21\,291, 2011.

\bibitem{montanari2010spread_innovation}
A.~Montanari and A.~Saberi, ``The spread of innovations in social networks,''
  \emph{Proc. Natl. Acad. Sci. USA}, vol. 107, no.~47, pp. 20\,196--20\,201,
  2010.

\bibitem{riehl2018survey}
J.~Riehl, P.~Ramazi, and M.~Cao, ``A survey on the analysis and control of
  evolutionary matrix games,'' \emph{Annu. Rev. Control}, vol.~45, pp. 87--106,
  2018.

\bibitem{flache2017opdyn_survey}
A.~Flache, M.~M{\"a}s, T.~Feliciani, E.~Chattoe-Brown, G.~Deffuant, S.~Huet,
  and J.~Lorenz, ``{Models of Social Influence: Towards the Next Frontiers},''
  \emph{J. Artif. Soc. S.}, vol.~20, no.~4, 2017.

\bibitem{ye2021nat}
M.~Ye \emph{et~al.}, ``Collective patterns of social diffusion are shaped by
  individual inertia and trend-seeking,'' \emph{Nat. Comm.}, vol.~12, p. 5698,
  2021.

\bibitem{henry2003beyond}
G.~T. Henry and M.~M. Mark, ``Beyond use: Understanding evaluation's influence
  on attitudes and actions,'' \emph{Am. J. Eval.}, vol.~24, no.~3, pp.
  293--314, 2003.

\bibitem{haidt2001emotional}
J.~Haidt, ``The emotional dog and its rational tail: a social intuitionist
  approach to moral judgment.'' \emph{Psychol. Rev.}, vol. 108, no.~4, p. 814,
  2001.

\bibitem{gavrilets2017collective}
S.~Gavrilets and P.~J. Richerson, ``Collective action and the evolution of
  social norm internalization,'' \emph{Proc. Natl. Acad. Sci. USA}, vol. 114,
  no.~23, pp. 6068--6073, 2017.

\bibitem{martins2009opinion}
A.~C. Martins, C.~d.~B. Pereira, and R.~Vicente, ``An opinion dynamics model
  for the diffusion of innovations,'' \emph{Physica A}, vol. 388, no. 15-16,
  pp. 3225--3232, 2009.

\bibitem{martins2008continuous}
A.~C. Martins, ``Continuous opinions and discrete actions in opinion dynamics
  problems,'' \emph{Int. J. Mod. Phys. C}, vol.~19, no.~04, pp. 617--624, 2008.

\bibitem{chowdhury2016}
N.~R. Chowdhury, I.-C. Morărescu, S.~Martin, and S.~Srikant, ``Continuous
  opinions and discrete actions in social networks: A multi-agent system
  approach,'' in \emph{IEEE 55th Conf. Decis. Control}, 2016, pp. 1739--1744.

\bibitem{ceragioli2018}
F.~Ceragioli and P.~Frasca, ``Consensus and disagreement: The role of quantized
  behaviors in opinion dynamics,'' \emph{SIAM J. Control Optim.}, vol.~56,
  no.~2, pp. 1058--1080, 2018.

\bibitem{varma2017}
V.~S. Varma, I.-C. Morărescu, S.~Lasaulce, and S.~Martin, ``Opinion dynamics
  aware marketing strategies in duopolies,'' in \emph{IEEE 56th Conf. Decis.
  Control}, 2017, pp. 3859--3864.

\bibitem{zino2020two}
L.~Zino, M.~Ye, and M.~Cao, ``A two-layer model for coevolving opinion dynamics
  and collective decision-making in complex social systems,'' \emph{Chaos},
  vol.~30, no.~8, p. 083107, 2020.

\bibitem{zino2020coevolutionary}
------, ``A coevolutionary model for actions and opinions in social networks,''
  in \emph{59th IEEE Conf. Decis. Control.}, 2020, pp. 1110--1115.

\bibitem{aghbolagh2023coevolutionary}
H.~D. Aghbolagh, M.~Ye, L.~Zino, Z.~Chen, and M.~Cao, ``Coevolutionary dynamics
  of actions and opinions in social networks,'' \emph{IEEE Trans. Autom.
  Control.}, vol.~68, no.~12, pp. 7708--7723, 2023.

\bibitem{monderer1996potential}
D.~Monderer and L.~S. Shapley, ``Potential games,'' \emph{Games Econ. Behav.},
  vol.~14, no.~1, pp. 124--143, 1996.

\bibitem{Bramoull2007}
Y.~Bramoullé, ``Anti-coordination and social interactions,'' \emph{Games Econ.
  Behav.}, vol.~58, no.~1, pp. 30--49, 2007.

\bibitem{fudenberg1991game}
D.~Fudenberg and J.~Tirole, \emph{Game theory}.\hskip 1em plus 0.5em minus
  0.4em\relax MIT Press, 1991.

\bibitem{alvarez2004network}
F.~Alvarez and K.~Ritzberger, ``Network formation and anti-coordination
  games,'' \emph{Int. J. Game Theory}, vol.~32, pp. 501--517, 2004.

\bibitem{ramazi2016networks}
P.~Ramazi, J.~Riehl, and M.~Cao, ``Networks of conforming or nonconforming
  individuals tend to reach satisfactory decisions,'' \emph{Proc. Natl. Acad.
  Sci. USA}, vol. 113, no.~46, pp. 12\,985--12\,990, 2016.

\bibitem{Vanelli2020}
M.~Vanelli, L.~Arditti, G.~Como, and F.~Fagnani, ``On games with coordinating
  and anti-coordinating agents,'' in \emph{21st IFAC World Cong.}, vol.~53,
  no.~2.\hskip 1em plus 0.5em minus 0.4em\relax Elsevier BV, 2020, pp.
  10\,975--10\,980.

\bibitem{degroot1974reaching}
M.~H. DeGroot, ``Reaching a consensus,'' \emph{J. Am. Stat. Assoc.}, vol.~69,
  no. 345, pp. 118--121, 1974.

\bibitem{marden2009game}
J.~R. {Marden}, G.~{Arslan}, and J.~S. {Shamma}, ``Cooperative control and
  potential games,'' \emph{IEEE Trans. Syst. Man Cybern. B}, vol.~39, no.~6,
  pp. 1393--1407, 2009.

\bibitem{friedkin1990_FJsocialmodel}
N.~E. Friedkin and E.~C. Johnsen, ``Social influence and opinions,'' \emph{J.
  Math. Sociol.}, vol.~15, no. 3-4, pp. 193--206, 1990.

\bibitem{Takcs2016}
K.~Takács, A.~Flache, and M.~M\"{a}s, ``Discrepancy and disliking do not
  induce negative opinion shifts,'' \emph{PLOS ONE}, vol.~11, no.~6, p.
  e0157948, 2016.

\bibitem{mas2016behavioral}
M.~M{\"a}s and H.~H. Nax, ``A behavioral study of noise in coordination
  games,'' \emph{J. Econ. Theory}, vol. 162, pp. 195--208, 2016.

\bibitem{xiao2006state}
F.~Xiao and L.~Wang, ``State consensus for multi-agent systems with switching
  topologies and time-varying delays,'' \emph{Int. J. Control}, vol.~79,
  no.~10, pp. 1277--1284, 2006.

\bibitem{raineri2025}
R.~Raineri, G.~Como, F.~Fagnani, M.~Ye, and L.~Zino, ``On controlling a
  coevolutionary model of actions and opinions,'' in \emph{IEEE 63rd Conf.
  Decis. Control}, 2024, pp. 4550--4555.

\end{thebibliography}

\begin{IEEEbiography}[{\includegraphics[width=1in,height=1.25in,clip,keepaspectratio]{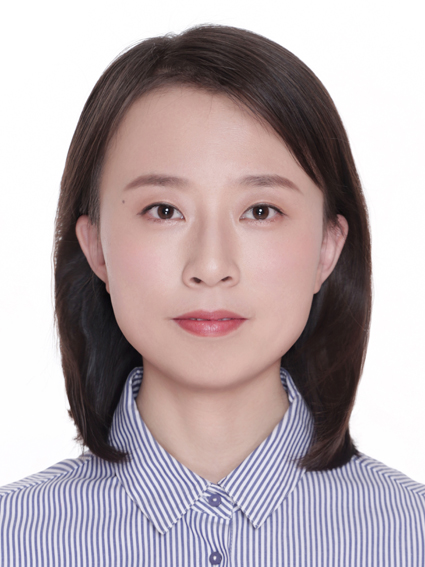}}]{Hong Liang} received the B.S. degree in basic mathematics from Inner Mongolia University, China, in 2016 and the M.S. degree in operations research and control from Dalian University of Technology, China, in 2019. She is currently working toward the Ph.D. degree in control engineering with the Dalian University of Technology, China. Her current research interests include multi-agent systems, game theory, and social networks.
\end{IEEEbiography}

\begin{IEEEbiography}[{\includegraphics[width=1in,height=1.25in,clip,keepaspectratio]{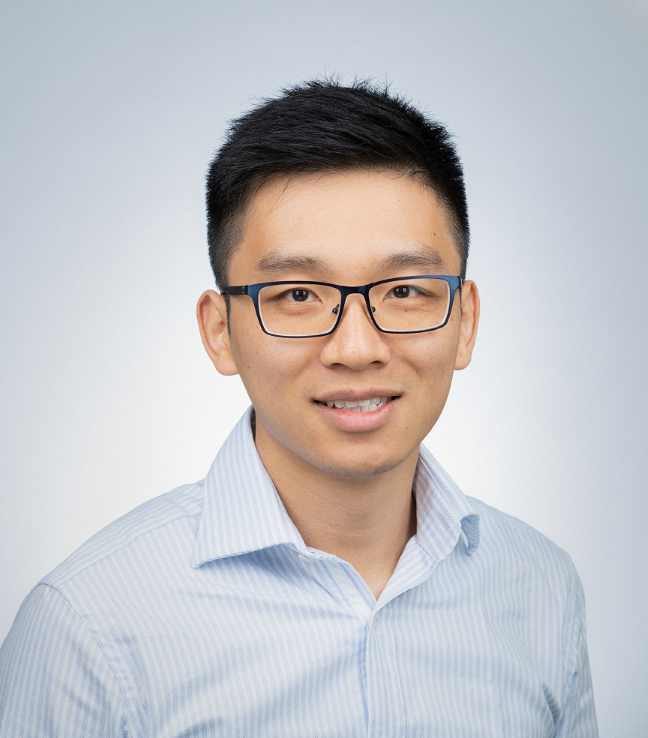}}]{Mengbin Ye} received the BE degree (with First Class Honours) in mechanical engineering from the University of Auckland, New Zealand (2013) and the PhD degree in engineering at the Australian National University, Australia (2018). From 2018-20, he was a postdoctoral researcher with the Faculty of Science and Engineering, University of Groningen, Netherlands. He joined Curtin University, Australia, in 2020 and held a four-year Western Australia Premier's Early Career Fellowship (2021-25). In 2025, he joined the University of Adelaide, Australia, as a Senior Research Fellow, funded by an Australian Research Council Discovery Early Career Research Award (DECRA). He was awarded the J.G. Crawford Prize (Interdisciplinary) in 2018, ANU’s premier award recognizing graduate research excellence, and received the 2018 Springer PhD Thesis Prize. His current research interests include dynamics on social networks, epidemic modeling and control, and cooperative control of multi-agent systems.
\end{IEEEbiography}

\begin{IEEEbiography}[{\includegraphics[width=1in,height=1.25in,clip,keepaspectratio]{ 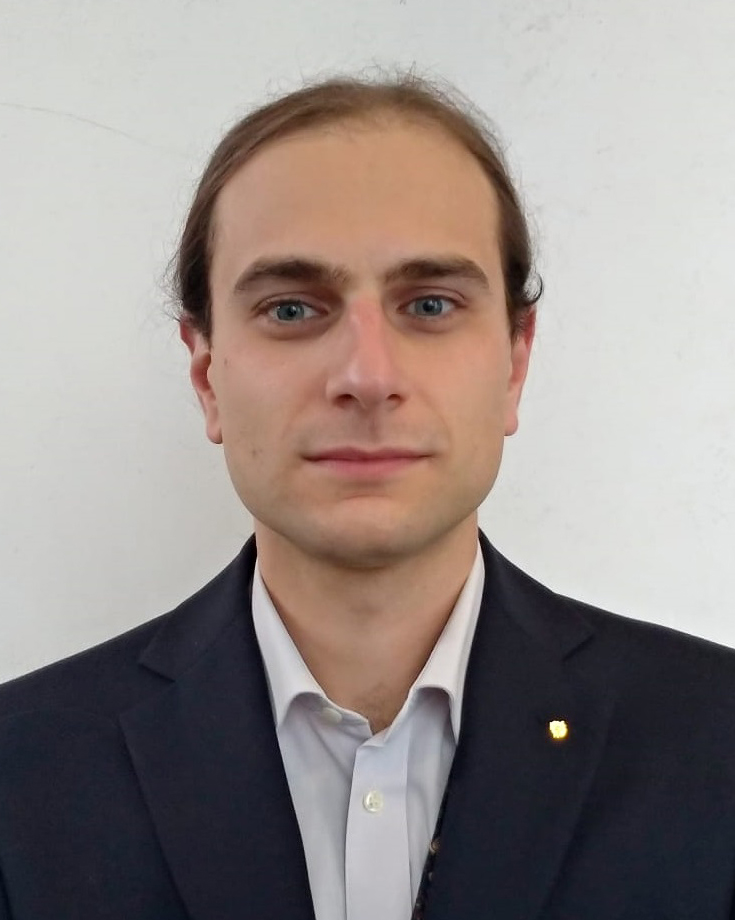}}]{Lorenzo Zino}  is an Assistant Professor with the Department of Electronics and Telecommunications, Politecnico di Torino, Italy, since 2022. He  received BS (2012), MS (summa cum laude, 2014), and PhD (with honors, 2018) in Applied Mathematics  from  Politecnico  di  Torino. He held Research Fellowships at Politecnico di Torino, University of Groningen, and New York University. His research interests include modeling, analysis, and control of network dynamics,  and game theory. He has co-authored more than 80 scientific publications, including more than 50 journal papers. 
In 2024, he was the recipient of the Best Young Author Journal Paper Award from the  IEEE CSS Italy Chapter. He is member of the Editorial Board of \emph{Scientific Reports}, Associate Editor of the \emph{Journal of Computational Science}, and member of the CEB for IEEE CSS and EUCA.
\end{IEEEbiography}

\begin{IEEEbiography}[{\includegraphics[width=1in,height=1.25in, clip,keepaspectratio]{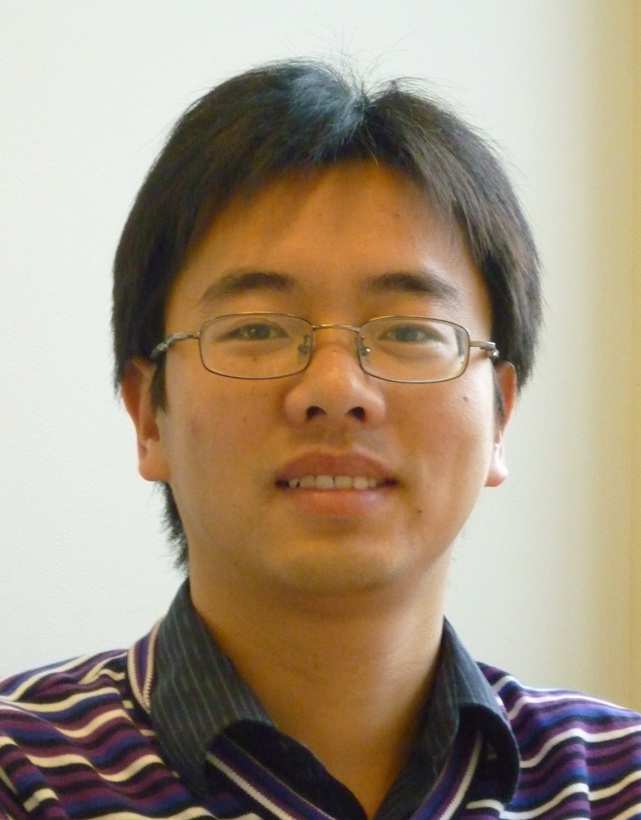}}]{Weiguo Xia}   received the B.Sc. and M.Sc. degrees in applied mathematics from Southeast University, Nanjing, China, in 2006 and 2009, respectively, and the Ph.D. degree in systems and control from the Faculty of Science and
Engineering, University of Groningen,Groningen, The Netherlands, in 2013.He is currently a Professor with the School of Control Science and Engineering, Dalian University of Technology, Dalian, China. From 2013 to 2015, he was a Postdoctoral Researcher with the ACCESS Linnaeus Centre, Royal Institute of Technology, Stockholm, Sweden. His research interests include complex networks, social networks, and multiagent systems. He is Associate Editor for Systems \& Control Letters.
\end{IEEEbiography}

\end{document}